\documentclass{article}
\usepackage[T1]{fontenc}
\usepackage[latin9]{inputenc}
\usepackage{geometry}
\usepackage{amsthm}
\usepackage{amsmath}
\usepackage{amssymb}

\makeatletter
\theoremstyle{plain}
\newtheorem{thm}{\protect\theoremname}[section]
\theoremstyle{definition}
\newtheorem{defn}[thm]{\protect\definitionname}
\theoremstyle{remark}
\newtheorem{rem}[thm]{\protect\remarkname}
\theoremstyle{plain}
\newtheorem{prop}[thm]{\protect\propositionname}
\ifx\proof\undefined
\newenvironment{proof}[1][\protect\proofname]{\par
\normalfont\topsep6\p@\@plus6\p@\relax
\trivlist
\itemindent\parindent
\item[\hskip\labelsep
\scshape
#1]\ignorespaces
}{%
\endtrivlist\@endpefalse
}
\providecommand{\proofname}{Proof}
\fi
\theoremstyle{plain}
\newtheorem{lem}[thm]{\protect\lemmaname}
\theoremstyle{definition}
\newtheorem{example}[thm]{\protect\examplename}
\theoremstyle{plain}
\newtheorem{cor}[thm]{\protect\corollaryname}

\RequirePackage{fix-cm}
\usepackage{ae,aecompl}  

\usepackage[]{hyperref}
\usepackage{url}

\usepackage{cite}
\usepackage[title,titletoc]{appendix}

\date{}  

\makeatother

\usepackage[english]{babel}
\providecommand{\corollaryname}{Corollary}
\providecommand{\definitionname}{Definition}
\providecommand{\examplename}{Example}
\providecommand{\lemmaname}{Lemma}
\providecommand{\propositionname}{Proposition}
\providecommand{\remarkname}{Remark}
\providecommand{\theoremname}{Theorem}

\begin{document}

\title{Fundamental theorems of asset pricing for piecewise semimartingales
of stochastic dimension\global\long\global\long\global\long\def\norm#1{\left\Vert #1\right\Vert }

\global\long\global\long\global\long\def\abs#1{\left\vert #1\right\vert }

\global\long\global\long\global\long\def\set#1{\left\{  #1\right\}  }

\global\long\global\long\global\long\def\eps{\varepsilon}

\global\long\global\long\global\long\def\zero{\odot}

\global\long\global\long\global\long\def\fa{\mathfrak{a}}

\global\long\global\long\global\long\def\cA{\mathcal{A}}

\global\long\global\long\global\long\def\cB{\mathcal{B}}

\global\long\global\long\global\long\def\fb{\mathfrak{b}}

\global\long\global\long\global\long\def\C{\mathbb{C}}

\global\long\global\long\global\long\def\bC{\mathbb{C}}

\global\long\global\long\global\long\def\cC{\mathcal{C}}

\global\long\global\long\global\long\def\cD{\mathcal{D}}

\global\long\global\long\global\long\def\bD{\mathbb{D}}

\global\long\global\long\global\long\def\rd{\mathrm{d}}

\global\long\global\long\global\long\def\d{\mathrm{d}}

\global\long\global\long\global\long\def\cE{\mathcal{E}}

\global\long\global\long\global\long\def\cF{\mathcal{F}}

\global\long\global\long\global\long\def\bF{\mathbb{F}}

\global\long\global\long\global\long\def\cG{\mathcal{G}}

\global\long\global\long\global\long\def\fG{\mathfrak{G}}

\global\long\global\long\global\long\def\fg{\mathfrak{g}}

\global\long\global\long\global\long\def\bG{\mathbb{G}}

\global\long\global\long\global\long\def\bH{\mathbb{H}}

\global\long\global\long\global\long\def\cH{\mathcal{H}}

\global\long\global\long\global\long\def\cI{\mathcal{I}}

\global\long\global\long\global\long\def\cK{\mathcal{K}}

\global\long\global\long\global\long\def\bK{\mathbb{K}}

\global\long\global\long\global\long\def\bL{\mathbb{L}}

\global\long\global\long\global\long\def\cL{\mathcal{L}}

\global\long\global\long\global\long\def\cM{\mathcal{M}}

\global\long\global\long\global\long\def\fm{\mathfrak{m}}

\global\long\global\long\global\long\def\N{\mathbb{N}}

\global\long\global\long\global\long\def\bN{\mathbb{N}}

\global\long\global\long\global\long\def\cN{\mathcal{N}}

\global\long\global\long\global\long\def\cO{{\normalcolor \mathcal{O}}}

\global\long\global\long\global\long\def\bP{\mathbb{P}}

\global\long\global\long\global\long\def\cP{\mathcal{P}}

\global\long\global\long\global\long\def\fP{\mathfrak{P}}

\global\long\global\long\global\long\def\fp{\mathfrak{p}}

\global\long\global\long\global\long\def\bQ{\mathbb{Q}}

\global\long\global\long\global\long\def\cQ{\mathcal{Q}}

\global\long\global\long\global\long\def\fr{\mathfrak{r}}

\global\long\global\long\global\long\def\R{\mathbb{R}}

\global\long\global\long\global\long\def\bR{\mathbb{R}}

\global\long\global\long\global\long\def\cR{\mathcal{R}}

\global\long\global\long\global\long\def\fR{\mathfrak{R}}

\global\long\global\long\global\long\def\bS{\mathbb{S}}

\global\long\global\long\global\long\def\cS{\mathcal{S}}

\global\long\global\long\global\long\def\fs{\mathfrak{s}}

\global\long\global\long\global\long\def\cT{\mathcal{T}}

\global\long\global\long\global\long\def\bT{\mathbb{T}}

\global\long\global\long\global\long\def\ft{\mathfrak{t}}

\global\long\global\long\global\long\def\bU{\mathbb{U}}

\global\long\global\long\global\long\def\cU{\mathcal{U}}

\global\long\global\long\global\long\def\bV{\mathbb{V}}

\global\long\global\long\global\long\def\cV{\mathcal{V}}

\global\long\global\long\global\long\def\cX{\mathcal{X}}

\global\long\global\long\global\long\def\fX{\mathfrak{X}}

\global\long\global\long\global\long\def\cY{\mathcal{Y}}

\global\long\global\long\global\long\def\fY{\mathfrak{Y}}

\global\long\global\long\global\long\def\fy{\mathfrak{y}}

\global\long\global\long\global\long\def\bZ{\mathbb{Z}}

\global\long\global\long\global\long\def\cZ{\mathcal{Z}}

\global\long\global\long\global\long\def\fZ{\mathfrak{Z}}

\global\long\global\long\global\long\def\I{\mathbf{1}}

\global\long\global\long\global\long\def\cemetery{\dagger}

\global\long\global\long\global\long\def\D#1#2{\frac{\partial#1}{\partial#2}}

\global\long\global\long\global\long\def\DD#1#2{\frac{\partial^{2}#1}{\partial#2^{2}}}

$\global\long\global\long\global\long\def\vec#1{\mbox{\boldmath\ensuremath{#1}}}
$

\global\long\global\long\global\long\def\wt#1{\widetilde{#1}}

\global\long\global\long\global\long\def\1{\mathbf{1}}

\global\long\global\long\global\long\def\uI{\mathbf{\hat{1}}}

\global\long\global\long\global\long\def\2{\mathbf{1}}

\global\long\global\long\global\long\def\asto{\xrightarrow{\text{a.s.}}}

\global\long\global\long\global\long\def\Lto{\xrightarrow{L^{1}}}

\global\long\global\long\global\long\def\Lpto{\xrightarrow{L^{p}}}

\global\long\global\long\global\long\def\asLto{\xrightarrow{L^{1}, \text{ a.s.}}}

\global\long\global\long\global\long\def\imply{\Rightarrow}

\global\long\global\long\global\long\def\nimply{\nRightarrow}

\global\long\global\long\global\long\def\limply{\Longrightarrow}

\global\long\global\long\global\long\def\leftexp#1#2{{{\vphantom{#2}}^{#1}{#2}}}

\global\long\global\long\global\long\def\var{\textrm{Var}}

\global\long\global\long\global\long\def\std{\textrm{Std}}

\global\long\global\long\global\long\def\corr{\textrm{Corr}}

\global\long\global\long\global\long\def\cov{\textrm{Cov}}

\global\long\global\long\global\long\def\tr{\textrm{tr}}

\global\long\global\long\global\long\def\diag{\textrm{diag}}

\global\long\global\long\global\long\def\sgn{\textrm{sign}}

\global\long\global\long\global\long\def\argmin{\operatornamewithlimits{arg\, min}}

\global\long\global\long\global\long\def\argmax{\operatornamewithlimits{arg\, max}}

\global\long\global\long\global\long\def\bbmid{\big|}

\global\long\global\long\global\long\def\bmid{\textrm{\ensuremath{\Big|}}}

\global\long\global\long\global\long\def\Bmid{\bigg|}

\global\long\global\long\global\long\def\BBmid{\Bigg|}

}

\author{Winslow Strong%
\thanks{I gratefully acknowledge financial support by the National Centre
of Competence in Research ``Financial Valuation and Risk Management''
(NCCR FINRISK), Project D1 (Mathematical Methods in Financial Risk
Management). The NCCR FINRISK is a research instrument of the Swiss
National Science Foundation. Additionally, financial support by the
ETH Foundation and a Chateaubriand fellowship, from the Embassy of
France in the United States, are gratefully acknowledged. %
}}

\maketitle
\noindent \begin{center}
ETH Zurich, Department of Mathematics
\par\end{center}

\noindent \begin{center}
CH-8092 Zurich, Switzerland
\par\end{center}

\noindent \begin{center}
winslow.strong@math.ethz.ch
\par\end{center}

\smallskip{}

\begin{center}
\emph{First Draft: }12/11, \emph{Current Draft: }12/11
\par\end{center}
\begin{abstract}
The purpose of this paper is two-fold. First is to extend the notions
of an $n$-dimensional semimartingale and its stochastic integral
to a \emph{piecewise semimartingale} \emph{of stochastic dimension}.
The properties of the former carry over largely intact to the latter,
avoiding some of the pitfalls of infinite-dimensional stochastic integration.
Second is to extend two fundamental theorems of asset pricing (FTAPs):
the equivalence of no free lunch with vanishing risk to the existence
of an equivalent sigma-martingale measure for the price process, and
the equivalence of no arbitrage of the first kind to the existence
of an equivalent local martingale deflator for the set of nonnegative
wealth processes.
\end{abstract}
\medskip{}

\noindent \begin{flushleft}
\par\end{flushleft}{\small \par}

\noindent \begin{flushleft}
\par\end{flushleft}

\noindent \begin{flushleft}
\par\end{flushleft}

\section{Introduction and background\label{Sec:Introduction}}

\subsection{Piecewise semimartingales\label{Sub:Piecewise_Semis}}

Stochastic processes with finite, stochastic dimension have been studied
previously, for example in the theory of branching processes and diffusions,
but it does not appear that a general theory of stochastic integration
with respect to these processes has been developed. This case lies
in-between that of infinite-dimensional stochastic integration and
the fixed-finite-dimensional case. The stronger properties of the
latter carry over largely intact to the case of stochastic dimension.
This is one reason for our choice of developing the theory by extending
finite-dimensional stochastic integration via localization, rather
than treating it as a special case of infinite-dimensional stochastic
integration. The other main reason for this approach is that the finite-dimensional
treatment is more elementary, and therefore accessible to a broader
audience.

\subsubsection{Related notions of stochastic integration\label{SubSub:Related_notions_of_Stoch_Int}}

Stochastic integration has previously been extended to integrators
taking values in infinite-dimensional spaces of varying generality
\cite{Art:BjorkEtAl:TowardsGenTheorBondMarks:1997,Book:CarmonaTehranchi:IntRateModelsInfDimStochAnal:2006,Art:Mikulevicius:StochIntTVS:1998,Art:Mikulevicius:MartProbStochPDEs:1999}.
The case that is closest to that of finite-dimensional semimartingale
integration is when the integrator is a sequence of semimartingales,
as developed by De Donno and Pratelli \cite{Art:DeDonnoPratelli:StochIntSeqSemis:2006}.
Their formulation preserves many, but not all, of the nice properties
of finite-dimensional stochastic integration. For example the Ansel
and Stricker theorem \cite{Art:AnselStricker:Thm:1994}, which gives
necessary and sufficient conditions for $H\cdot X$ to be a local
martingale when $X$ is an $\R^{n}$-valued local martingale, does
not extend. A counterexample is given as Example $2$ in \cite{Art:DeDonnoPratelli:StochIntSeqSemis:2006}
where $(H\cdot X)_{t}=t$, $\forall t\ge0$, with $X$ a local martingale.

This pathology presents a difficulty for defining admissibility of
trading strategies in market models where the price process is a sequence
of semimartingales. The notion of a limited credit line ($H\cdot X$
uniformly bounded from below) is no longer sufficient to rule out
arbitrage. Instead, more technical formulations of admissibility are
necessary \cite{Art:DeDonnoGuasoniPratelli:SuperRepUtilMaxLargeFinMarks:2005}.
However, the theory of stochastic integration with respect to piecewise
semimartingales, developed herein, does not have such problems. The
Ansel and Sticker theorem extends as Theorem \ref{Thm:Ansel_Stricker},
and consequently if $H\cdot X$ is uniformly bounded from below, then
$H$ is admissible.

\subsubsection{Piecewise integration\label{SubSec:Piecewise_Integration}}

The theory of stochastic integration developed herein is a piecewise
one. The integrand $X$ takes values in $\cup_{n=1}^{\infty}\R^{n}$,
and its integral is formed by \emph{dissection, }that is, by localization
on stochastic time intervals $(\tau_{k-1},\tau_{k}]$ and partitioning
on the dimension of the integrator. Then the stochastic integrals
with respect to $X^{k,n}$, the $\R^{n}$-valued semimartingale ``pieces''
of $X$, may be stitched together to define $H\cdot X:=H_{0}^{\prime}X_{0}+\sum_{k,n=1}^{\infty}H^{k,n}\cdot X^{k,n}$,
and $H^{\prime}$ denotes the transpose of $H$.

This notion of piecewise integration provides one possible solution
for how to deal with integration over dimensional changes. In $\R^{n}$-valued
semimartingale stochastic integration, $X$ is assumed to have right-continuous
paths, and $\Delta(H\cdot X)=H^{\prime}\Delta X$, where $\Delta X:=X-X^{-}$,
and $X^{-}$ is the left-limit process of $X$. However, since $x-y$
is undefined when $\dim x\ne\dim y$, for $x,y\in\cup_{n=1}^{\infty}\R^{n}$,
this approach does not immediately extend to dimensional shifts. One
solution would be to adopt the convention of treating nonexistent
components as if they take the value $0$ (similar to the convention
of $(\Delta(H\cdot X))_{0}:=H_{0}^{\prime}X_{0}$ in $\R^{n}$-stochastic
integration, as in \cite{Book:Protter:SDE:2005}).

However, here we take a different perspective, and place primary importance
on preserving $H\cdot X$ as the capital gains (profits) arising from
holding $H$ shares in risky assets with price process $X$. This
is mainly due to the naturalness of $H\cdot X$ in this role and the
centrality of capital gains to financial mathematics. For example,
when a new asset enters the investable universe, its mere existence
as an option for investment does not cause any portfolio values to
change. So portfolio values should be \emph{conserved }upon such an
event, making the jump notion considered above incompatible with maintaining
$H\cdot X$ as the capital gains process.

Instead, dimensional jumps in $X$ are mandated to occur only as right
discontinuities. This allows stochastic integration to be stopped
just before each jump and resumed just afterwards. The left discontinuities,
as usual, influence $H\cdot X$, while the right discontinuities serve
to indicate the start of a new piece and do not affect $H\cdot X$,
which remains a right-continuous process. Having decided on this convention
for how to handle the dimensional jumps, there is no reason to restrict
those jumps of $X$ that do not influence $H\cdot X$ to merely changes
in dimension. Hence, $X$ is permitted to have right discontinuities
without changing dimension. Since this piecewise procedure will fail
to define a process on $\R_{+}\times\Omega$ if infinitely many pieces
are required on a compact time interval, the paths of $X$ are required
to have no accumulation points of right discontinuities.

\subsection{Fundamental theorems of asset pricing\label{Sub:Fundamental-theorems-of}}

There has been no small amount of literature on the topic of FTAPs
in different settings. We do not attempt to provide a full history
here, for which the interested reader should see \cite{Book:DelbSchach:ArbBook:2006}.
Instead, we  highlight some of the most important results pertaining
to the cases studied herein.

The paper by Delbaen and Schachermayer \cite{Art:DelbSchach:FundThmMathFin:1998}
proves the equivalence of the condition no free lunch with vanishing
risk (NFLVR)\emph{ }to the existence of an equivalent sigma-martingale
measure (E$\sigma$MM) for the price process $X$, when $X$ is an
$\R^{n}$-valued semimartingale. The special case where $X$ is locally
bounded yields the existence of an equivalent local martingale measure
(ELMM) for $X$, a result proved earlier in \cite{Art:DelbSchach:FundThmMathFin:1994}.

For unbounded $X$, the paper of Kabanov \cite{Art:Kabanov:OnFTAPKrepsDelbSchach:1997}
concurrently arrived at the weaker equivalence of NFLVR with the existence
of an equivalent separating measure for the set of replicable claims.
However, his approach for this weaker result is more general than
\cite{Art:DelbSchach:FundThmMathFin:1998}, in that it merely requires
a closedness property of the replicable claims and some other basic
properties, rather than also imposing that these claims arise from
stochastic integration with respect to a semimartingale. Therefore,
his approach is well-suited for more general investigations into arbitrage,
including the case herein of $X$ as a piecewise semimartingale of
stochastic dimension. It is used in Section \ref{Sec:MarketModelsStochNumAssets},
along with Delbaen and Schachermayer's result \cite{Art:DelbSchach:FundThmMathFin:1998}
of E$\sigma$MMs being dense amongst the space of equivalent separating
measures, in order to prove Theorem \ref{Thm:FTAP}, a generalization
of $\mbox{NFLVR}\Longleftrightarrow\mbox{E}\sigma\mbox{MM}$ in the
piecewise setting.  Specializations are proved additionally, showing
that the sigma-martingale measures are local martingale measures when
the price process is locally bounded, in analogy with \cite{Art:DelbSchach:FundThmMathFin:1994}.

It does not appear that in the setting of infinite-dimensional stochastic
integration any sigma-martingale equivalence to a form of no approximate
arbitrage has been proved in the literature. A related result is proved
in \cite{Art:BalbasDownarowicz:InfManySecFTAP2007}, where the setting
is discrete time and the number of assets is countable, but the FTAP
does not extend in its original (discrete-time) form. Instead, no
arbitrage is characterized by projective limits of projective systems
of martingale measures.

\subsubsection{Large financial markets\label{SubSub:Large-financial-markets}}

The setting of large financial markets, introduced by Kabanov and
Kramkov in \cite{Art:KabanovKramkov:LargeFinMarks:AsympArb:1994},
bears resemblance to the setting herein, but is somewhat different,
since it consists of sequences of finite-dimensional market models
without the dynamics of a stochastic number of assets. They proved
FTAPs stating that for no asymptotic arbitrage of the second or first
kind to exist, respectively, it is sufficient that there exists a
sequence of local martingale measures for the finite-dimensional price
processes, contiguous with respect to the physical measures, or the
physical measures contiguous with respect to it, respectively. Furthermore,
they showed that if the market is complete, then this is necessary
as well. The completeness assumption was later shown to be unnecessary
concerning arbitrage of the first kind, by Klein and Schachermayer
in \cite{Art:Klein_Schachermayer:Asympt_Arb_Incomp_Larg_Fin_Mark:1997},
with alternative proofs in \cite{Art:Klein_Schachermayer:Halmos-Savage:1996,Art:KabanovKramkov:AsympArbLargFinMarks:1998}.
An equivalence of no asymptotic free lunch with the existence of a
bicontinuous sequence of sigma-martingale measures was additionally
proved by Klein in \cite{Art:Klein:FTAPforLargeFinMarks:2000}. There
it was also shown that no weakening to ``vanishing risk'' or even
``bounded risk'' is possible.

The paper of De Donno, Guasoni, and Pratelli, \cite{Art:DeDonnoGuasoniPratelli:SuperRepUtilMaxLargeFinMarks:2005},
studies super-replication and utility maximization using duality methods
in a market modeled by a sequence of semimartingales, using the integration
theory developed in \cite{Art:DeDonnoPratelli:StochIntSeqSemis:2006}.
Sequences of finite-dimensional markets are used as approximations
to the countable-asset market, and it is assumed that a measure exists
that martingalizes the entire sequence of asset prices, a stronger
condition than in the previously mentioned results.

\subsubsection{No arbitrage of the first kind\label{SubSub:No-arbitrage-of-first-kind}}

A different FTAP is also proved here as Theorem \ref{Thm:NA1_ELMD},
proved by Kardaras in \cite[Theorem 1.1]{Art:Kardaras:MarkViabAbsArb1stKind:2012}
for the one-dimensional semimartingale case. It is a much weaker condition
than NFLVR, indeed even admitting certain types of arbitrage. The
statement is that no arbitrage of the first kind is equivalent to
the existence of an equivalent local martingale deflator (ELMD) for
the set of nonnegative wealth processes. Notably, this condition does
not require the closure property of passing from local martingales
to martingales, so it has the virtue of being verifiable via local
arguments, which is not the case for the NFLVR FTAP. When an ELMD
exists in $\R^{n}$-valued semimartingale markets, it provides sufficient
regularity for a duality-based theory of hedging and utility maximization
\cite{Art:Karatzas&Fernholz:SPTReview:2009,Art:Ruf:OptTradStratUnderArb:2009}.

\section{Piecewise semimartingales of stochastic dimension\label{Sec:Piecewise-semimartingales-of-Stoch-Dim}}

This section will motivate and develop the notion of a piecewise semimartingale
whose dimension is a finite but unbounded stochastic process, and
extend stochastic integration to these processes as integrators. A
natural%
\footnote{Another choice could be the space of sequences that have all but finitely
many terms equal to $0$. However, this state space lacks the dimensional
information of $\bU$. This information would need to be supplied
as an auxiliary process for the theory of integration developed herein.%
} state space for such a process is $\bU:=\cup_{n=1}^{\infty}\R^{n}$,
equipped with the topology generated by the union of the standard
topologies on each $\R^{n}$. When $x,y\in\R^{n}$, then $x+y$ is
defined as usual, and multiplication by a scalar is defined as usual
within each $\R^{n}$. For regularity considerations, we will limit
discussion to processes whose paths are composed of finitely many
càdlàg pieces on all compact time intervals. Each change in dimension
of the process necessitates the start of a new piece, so may only
occur at a right discontinuity.

\subsection{Notation\label{Sub:Notation}}

The basic technique for manipulating $\bU$-valued piecewise processes
will be\emph{ dissection},  meaning localization on stochastic time
intervals and partitioning into $\R^{n}$-valued processes. Then standard
results from $\R^{n}$-valued stochastic analysis can be applied and
extended.

Indicator functions are a useful notational tool for dissecting stochastic
processes, but must be reformulated to be useful in the state space
$\bU$, due to the multiplicity of zeros: $0^{(n)}\in\R^{n}$. To
salvage their utility, define an additive identity element $\zero$,
a topologically isolated point in $\widehat{\bU}:=\bU\cup\{\zero\}$,
distinguished from $0^{(n)}\in\R^{n}$, $n\in\N$. Let $\zero+x=x+\zero=x$,
and $\zero x=x\zero=\zero$, for each $x\in\widehat{\bU}$. The modified
indicator will be denoted by
\begin{align*}
\uI_{A}(t,\omega) & :=\begin{cases}
1\in\R, & \mbox{for }(t,\omega)\in A\subseteq[0,\infty)\times\Omega\\
\zero & \mbox{otherwise }
\end{cases}.
\end{align*}
The usual definition of indicator will still be useful, which will
be denoted by the usual notation: $\I_{A}$.

All relationships among random variables hold merely almost surely
(a.s.), and for stochastic processes $Y$ and $Z$, $Y=Z$ will mean
that $Y$ and $Z$ are indistinguishable. The notations $\R_{+}:=[0,\infty)$
and $B^{\prime}$ to denote the transpose of a matrix $B$ will be
used. A process $Y$ stopped at a random time $\alpha$ will be denoted
$Y^{\alpha}:=(Y_{\alpha\wedge t})_{t\ge0}$. For any process $Y$
possessing paths with right limits at all times, $Y^{+}$ will denote
the right-limit process. All $\R^{n}$-valued semimartingales will
be assumed to have right-continuous paths. The $\bU$-extension of
the $l_{p}$-norms, referred to here as \emph{local norms}, will also
be useful. Of course, these are not norms in $\bU$, since $\bU$
is not even a vector space.
\begin{align*}
\left|\cdot\right|_{p} & :\bU\to\R_{+}\\
\left|h\right|_{p} & :=\left(\sum_{i=1}^{n}h_{i}^{p}\right)^{1/p},\quad\mbox{for }h\in\R^{n},\; n\in\N,\; p\in[1,\infty),\\
|h|_{\infty} & :=\max_{1\le i\le n}|h_{i}|,\quad\mbox{for }h\in\R^{n},\; n\in\N.
\end{align*}

\subsection{Stochastic integration\label{Sub:Stochastic-integration}}

Let the stochastic basis $(\Omega,\cF,\bF:=(\cF_{t})_{t\ge0},P)$
satisfy the usual conditions of $\cF_{0}$ containing the $P$-null
sets and $\bF$ being right-continuous. Let $X$ be a $\bU$-valued
progressively measurable process whose paths have left and right limits
at all times. In particular, this implies that $N:=\dim X$ also has
paths with left and right limits at all times.
\begin{defn}
\label{Def:ResetSeq}A sequence of stopping times $(\tau_{k})$ is
called a \emph{reset sequence} for a progressive $\bU$-valued process
$X$ if for $P$-almost every $\omega$ all of the following hold:

\begin{enumerate}
\item \label{Def:ResetSeq:Item:1}$\tau_{0}(\omega)=0$, $\tau_{k-1}(\omega)\le\tau_{k}(\omega)$,
$\forall k\in\N$,  and $\lim_{k\to\infty}\tau_{k}(\omega)=\infty$;
\item \label{Def:ResetSeq:Item:2}$N_{t}(\omega)=N_{\tau_{k-1}}^{+}(\omega)$
for all $t\in(\tau_{k-1}(\omega),\tau_{k}(\omega)]$, for each $k\in\N$;
\item \label{Def:ResetSeq:Item:3}$t\mapsto X_{t}(\omega)$ is right-continuous
on $(\tau_{k-1}(\omega),\tau_{k}(\omega))$ for all $k\in\N$.
\end{enumerate}
\end{defn}
If $X$ has a reset sequence, then the minimal one (in the sense of
the fewest resets by a given time) is given by $\hat{\tau}_{0}:=0$,
\begin{align}
\hat{\tau}_{k} & :=\inf\left\{ t>\hat{\tau}_{k-1}\mid X_{t}^{+}\ne X_{t}\right\} ,\quad k\in\N.\label{Eq:Minimal_ResetSeqDef}
\end{align}
The existence of a reset sequence for $X$ is a necessary regularity
condition for the theory herein, whereas the choice of reset sequence
is inconsequential for most applications, a fact addressed below.
\begin{rem}
The progressive property of $(X,\bF)$ is important, because it guarantees
that $(\hat{\tau}_{k})$ is in fact a sequence of stopping times.
Due to $X$ lacking the usual path regularity of having left or right
continuity at all times, it may not be the case that $X$ is progressive
with respect to its natural filtration. For this reason, it may more
appropriate to think of $(X,\hat{K})$ as the fundamental process
being studied here, where $\hat{K}_{t}:=\sum_{k=1}^{\infty}\I_{t\ge\hat{\tau}_{k}}$,
$t\ge0$, specifies what ``piece'' is active at time $t$. If $(\hat{\tau}_{k})$
satisfies \ref{Def:ResetSeq:Item:1}-\ref{Def:ResetSeq:Item:3} of
Definition \ref{Def:ResetSeq}, then $X$ can be shown to be progressive
with respect to the natural filtration of $(X,\hat{K})$, since the
$\hat{\tau}_{k}$ are stopping times on this filtration.
\end{rem}

The first technique that we develop is extending stochastic integration
to $X$ as integrator. The idea is that when $X$ has a discontinuity
from the right, the stochastic integral will ignore it. Integration
occurs from $0$ up to and including $\tau_{1}$, at which point the
integral is pasted together with an integral beginning just after
$\tau_{1}$, and so on.

For a process $X$ as described above and a reset sequence $(\tau_{k})$
for that process, \emph{dissect }$X$ and $\Omega$ to obtain the
following:
\begin{align}
\Omega^{k,n} & :=\{\tau_{k-1}<\infty,\; N_{\tau_{k-1}}^{+}=n\}\subseteq\Omega,\quad\forall k,n\in\N,\label{Eq:Omega_k_n}\\
X^{k,n} & :=(X^{\tau_{k}}-X_{\tau_{k-1}}^{+})\uI_{(\tau_{k-1},\infty)\cap\R_{+}\times\Omega^{k,n}}+0^{(n)},\quad\forall k,n\in\N.\label{Eq:X^kn-Pieces}
\end{align}
Each $X^{k,n}$ is $\R^{n}$-valued, adapted, has càdlàg paths, and
is therefore optional.
\begin{defn}
\label{Def:UValSemi}A \emph{ piecewise semimartingale} $X$ is a
$\bU$-valued progressive process having paths with left and right
limits for all times and possessing a reset sequence $(\tau_{k})$,
such that $X^{k,n}$ is an $\R^{n}$-valued semimartingale for each
$k,n\in\N$.
\end{defn}
Proposition \ref{Prop:Well_Definition_of_Semis_and_Stoch_Int} will
show that this definition and the subsequent development are not sensitive
to the choice of reset sequence. That is, if they hold for a particular
reset sequence, then they hold for any. The definition allows the
full generality of $\R^{n}$-valued semimartingale stochastic integration
theory to be carried over to  piecewise semimartingales taking values
in $\bU$.

Let $X$ be a  piecewise semimartingale and $(\tau_{k})$ a reset
sequence such that the $X^{k,n}$ are semimartingales, for each $k,n\in\N$.
Let $H$ be a $\bU$-valued predictable process satisfying $\dim H=N$.
Dissect $H$ via
\begin{align}
H^{k,n}: & =H\uI_{(\tau_{k-1},\tau_{k}]\cap\R_{+}\times\Omega^{k,n}}+0^{(n)},\quad k,n\in\N.\label{Eq:H^kn-Def}
\end{align}
Each $H^{k,n}$ is $\R^{n}$-valued and predictable, since $H$ is
predictable and $(\tau_{k-1},\tau_{k}]\cap\R_{+}\times\Omega^{k,n}$
is a predictable set.
\begin{defn}
\label{Def:StochasticIntegral}For a  piecewise semimartingale $X$
and reset sequence $(\tau_{k})$ let
\begin{align*}
\cL(X) & :=\{H\mbox{ predictable}\mid\dim H=N,\mbox{ and }\mbox{\ensuremath{H^{k,n}\;}is \ensuremath{X^{k,n}}-integrable \ensuremath{\forall k,n\in\N}\},}\\
\cL_{0}(X) & :=\{H\in\cL(X)\mid H_{0}=0^{(N_{0})}\}.
\end{align*}
For $H\in\cL(X)$, the \emph{stochastic integral} $H\cdot X$ is defined
as
\begin{align}
H\cdot X & :=H_{0}^{\prime}X_{0}+\sum_{k=1}^{\infty}\sum_{n=1}^{\infty}(H^{k,n}\cdot X^{k,n}).\label{Eq:Def-U-Val-StochInt}
\end{align}

\end{defn}
This stochastic integral is a generalization of $\R^{n}$-valued semimartingale
stochastic integration: In that special case any sequence stopping
times increasing to infinity is a reset sequence for $X$. Denoting
traditional stochastic integration with respect to $\R^{n}$-valued
semimartingale $X$ by $\widehat{H\cdot X}$, we have
\begin{align*}
H\cdot X & =H_{0}^{\prime}X_{0}+\sum_{k=1}^{\infty}\sum_{m=1}^{\infty}\widehat{H^{k,m}\cdot X^{k,m}},\\
 & =H_{0}^{\prime}X_{0}+\sum_{k=1}^{\infty}[(\widehat{H\cdot X})^{\tau_{k}}-(\widehat{H\cdot X})^{\tau_{k-1}}],\\
 & =\widehat{H\cdot X}.
\end{align*}

The stochastic integral operator $\cdot X$ retains the property of
being a continuous linear operator on the appropriate analog of the
space of simple predictable processes. A $\bU$-valued process $H$
is called \emph{simple predictable} for dimensional process $N$ if
it satisfies $\dim H=N$ and is of the form
\begin{align}
H= & H_{0}\uI_{\{0\}\times\Omega}+\sum_{i=1}^{j}H_{i}\uI_{(\alpha_{i},\alpha_{i+1}]}+0^{(N)},\label{Eq:Def_Simple_Predictable}
\end{align}
where $0=\alpha_{1}\le\ldots\le\alpha_{j+1}<\infty$, are stopping
times, and $H_{i}\in\cF_{\alpha_{i}}$, for $1\le i\le j$. The class
of such processes will be denoted $\bS(N)$, and when topologized
with the topology of uniform convergence on compact time sets in probability
(ucp), the resulting space will be denoted $\bS_{\textrm{ucp}}(N)$.
Similarly denote $\bD_{\textrm{ucp}}$ as the space of adapted processes
with right-continuous paths bearing the ucp topology.

\begin{prop}
\label{Prop:Semis_Are_Cont_Integrators}If $X$ is a  piecewise semimartingale
with $N=\dim X$, then
\begin{align*}
X:\bS_{\textrm{\emph{ucp}}}(N) & \to\bD_{\textrm{\emph{ucp}}}\\
X(H) & =H\cdot X
\end{align*}
is a continuous linear operator. \end{prop}
\begin{proof}
Let $H$, $H^{i}\in\bS_{\textrm{ucp}}(N)$, $\forall i\in\N$, and
$\lim_{i\to\infty}H^{i}=H$ (all limits here are assumed to be ucp).
Then by dissecting $H^{i}$ as in (\ref{Eq:H^kn-Def}) to get $H^{k,n,i}$,
and interchanging stopping and ucp limits, we have $\lim_{i\to\infty}H^{k,n,i}=H^{k,n}$
for all $k,n\in\N$. Since the $X^{k,n}$ are semimartingales, then
$\lim_{i\to\infty}(H^{k,n,i}\cdot X^{k,n})=(H^{k,n}\cdot X^{k,n})$.
Thus, $(\forall T>0)(\forall k_{0}\in\N)(\forall n_{0}\in\N)(\forall\varepsilon>0)(\forall\delta>0)(\exists i_{0}\in\N)$
such that
\begin{align}
P\left(\sup_{0\le t\le T}\left|(H^{k,n,i}\cdot X^{k,n})_{t}-(H^{k,n}\cdot X^{k,n})_{t}\right|>\varepsilon\right) & <\delta,\label{Eq:Prop:Semis_Are_Cont_Integrators:1}
\end{align}
whenever $k\le k_{0}$, $n\le n_{0}$, and $i>i_{0}$. For arbitrary
$\varepsilon,\rho>0$, choose $k_{0}$ sufficiently large such that
$P(\tau_{k_{0}}\le T)<\rho$, and $n_{0}$ sufficiently large such
that $P(\cup_{n>n_{0},k\le k_{0}}\Omega^{k,n})<\rho$. Then pick $i_{0}$
sufficiently large such that (\ref{Eq:Prop:Semis_Are_Cont_Integrators:1})
is satisfied for $\delta=\rho/(n_{0}k_{0})$. This results in the
following estimate for any $i\ge i_{0}$, proving the claim.
\begin{align*}
P\biggl(\sup_{0\le t\le T}\Bigl|(H^{i}\cdot X)_{t} & -(H\cdot X)_{t}>\varepsilon\Bigr|\biggr)\\
 & \le P\left(\tau_{k_{0}}>T\bigcap\sup_{0\le t\le T}\left|(H^{i}\cdot X)_{t}-(H\cdot X)_{t}\right|>\varepsilon\right)+P(\tau_{k_{0}}\le T),\\
 & \le\sum_{k=1}^{k_{0}}\sum_{n=1}^{n_{0}}P\left(\sup_{0\le t\le T}\left|(H^{k,n,i}\cdot X^{k,n})_{t}-(H^{k,n}\cdot X^{k,n})_{t}\right|>\varepsilon\right)\\
 & \quad+P\biggl(\bigcup_{n>n_{0},k\le k_{0}}\Omega^{k,n}\biggr)+P(\tau_{k_{0}}\le T),\\
 & \le k_{0}n_{0}\frac{\rho}{k_{0}n_{0}}+\rho+\rho=3\rho.\tag*{\qedhere}
\end{align*}

\end{proof}
The following proposition shows that the choice of reset sequence
$(\tau_{k})$ carries no significance in the definitions of piecewise
semimartingale, $\cL(X)$, nor $H\cdot X$. The proof is simple, yet
tedious, and is therefore relegated to Appendix \ref{App:Proofs}.
\begin{prop}
\label{Prop:Well_Definition_of_Semis_and_Stoch_Int}Let $X$ be a
 piecewise semimartingale and $\tilde{X}^{k,n}$ be defined as in
(\ref{Eq:X^kn-Pieces}), but with respect to an arbitrary reset sequence
$(\tilde{\tau}_{k})$. Then $\tilde{X}^{k,n}$ is an $\R^{n}$-valued
semimartingale for each $k,n\in\N$. Furthermore, the class $\cL(X)$
and the process $H\cdot X$ are invariant with respect to the choice
of reset sequence used in their definitions.
\end{prop}
Next we give some basic regularity properties of the stochastic integral.
\begin{prop}
\label{Prop:StochIntChar}The following properties hold for piecewise
semimartingale $X$:
\begin{enumerate}
\item \label{Prop:StochIntChar:Is_A_Semi}The stochastic integral $H\cdot X$
is an $\R$-valued semimartingale;
\item \label{Prop:StochIntChar:Integrands_Form_VS}$\cL(X)$ is a vector
space. If $H,G\in\cL(X)$ then $H\cdot X+G\cdot X=(H+G)\cdot X$;

\item \label{Prop:StochIntChars:Stopping}If $X$ is a  piecewise semimartingale
and $\alpha$ is a stopping time, then $X^{\alpha}$ is a  piecewise
semimartingale, and $(X^{\alpha})^{k,n}=(X^{k,n})^{\alpha}$ for any
$k,n\in\N$. Furthermore, if $H\in\cL(X)$, then $H\uI_{[0,\alpha]}+0^{(N)}\in\cL(X)$,
$H\uI_{[0,\alpha]}+0^{(N^{\alpha})}\in\cL(X^{\alpha})$, and
\begin{align*}
(H\cdot X)^{\alpha} & =(\uI_{[0,\alpha]}H+0^{(N)})\cdot X=(\uI_{[0,\alpha]}H+0^{(N^{\alpha})})\cdot X^{\alpha}.
\end{align*}

\end{enumerate}
\end{prop}
\begin{proof}
\textcompwordmark{}

\emph{\ref{Prop:StochIntChar:Is_A_Semi}. }Define $\Omega^{k,0}:=\{\tau_{k-1}=\infty\}$,
$k\in\N$, and $\cA:=\{(k,n)\in\N\times\bZ_{+}\mid P(\Omega^{k,n})>0\}$.
Define the probability measures $P^{k,n}(A):=\frac{P(A\cap\Omega^{k,n})}{P(\Omega^{k,n})}$,
$\forall A\in\cF$, $\forall(k,n)\in\cA$. Then $\sum_{n=1}^{\infty}(H^{k,n}\cdot X^{k,n})$
is a $P^{k,j}$-semimartingale for each $(k,j)\in\cA$, since $H^{k,n}\cdot X^{k,n}=0$
on $(\Omega^{k,n})^{c}$. For each $k\in\N$, we may represent $P$
as $P(A)=\sum_{n:(n,k)\in\cA}P(\Omega^{k,n})P^{k,n}(A)$, $\forall A\in\cF$,
where $\sum_{n:(n,k)\in\cA}P(\Omega^{k,n})=1$. Theorem II.3 of \cite{Book:Protter:SDE:2005}
then implies that $\sum_{n=1}^{\infty}(H^{k,n}\cdot X^{k,n})$ is
a $P$-semimartingale, $\forall k\in\N$. The stopped process $(H\cdot X)^{\tau_{m}}=\sum_{k=1}^{m}\sum_{n=1}^{\infty}(H^{k,n}\cdot X^{k,n})$
is a semimartingale $\forall m\in\N$, since it is a finite sum of
semimartingales. A process that is locally a semimartingale is a semimartingale
(corollary of Theorem II.6 \cite{Book:Protter:SDE:2005}) so we are
done.

\emph{\ref{Prop:StochIntChar:Integrands_Form_VS}. }By (\ref{Eq:H^kn-Def}),
$(H+G)^{k,n}=H^{k,n}+G^{k,n}$, and so the property follows from its
$\R^{n}$-counterpart.

\emph{\ref{Prop:StochIntChars:Stopping}.} Any reset sequence $(\tau_{k})$
for $X$ is a reset sequence for $X^{\alpha}$. Let $\tilde{\Omega}^{k,n}:=\{\tau_{k-1}<\infty,N_{\tau_{k-1}^{+}}^{\alpha}=n\}$,
and dissect $X^{\alpha}$ to get
\begin{align*}
(X^{\alpha})^{k,n} & =(X^{\tau_{k}\wedge\alpha}-(X^{\alpha})_{\tau_{k-1}}^{+})(\uI_{(\tau_{k-1},\infty)\cap\R_{+}\times\tilde{\Omega}^{k,n}})+0^{(n)},\\
 & =(X^{\tau_{k}}-X{}_{\tau_{k-1}}^{+})^{\alpha}\uI_{(\tau_{k-1},\infty)\cap\R_{+}\times(\tilde{\Omega}^{k,n}\cap\{\alpha>\tau_{k-1}\})}+0^{(n)},\\
 & =(X^{\tau_{k}}-X{}_{\tau_{k-1}}^{+})^{\alpha}(\uI_{(\tau_{k-1},\infty)\cap\R_{+}\times\Omega^{k,n}}+0^{(n)})^{\alpha},\\
 & =(X^{k,n})^{\alpha},
\end{align*}
where we used that $\tilde{\Omega}^{k,n}\cap\{\alpha>\tau_{k-1}\}=\Omega^{k,n}\cap\{\alpha>\tau_{k-1}\}$,
and $(X^{\alpha})^{k,n}=(X^{k,n})^{\alpha}=0^{(n)}$ on $\{\alpha\le\tau_{k-1}\}$.
Since $X^{k,n}$ is a semimartingale, and semimartingales are stable
with respect to stopping, then $(X^{\alpha})^{k,n}$ is a semimartingale,
$\forall k,n\in\N$, so $X^{\alpha}$ is a  piecewise semimartingale.

If $H\in\cL(X)$, then $\dim(H\uI_{[0,\alpha]}+0^{(N^{\alpha})})=\dim X^{\alpha}=N^{\alpha}$,
and dissection yields
\begin{align*}
(H\uI_{[0,\alpha]}+0^{(N^{\alpha})})^{k,n} & =(H\uI_{[0,\alpha]}+0^{(N^{\alpha})})\uI_{(\tau_{k-1},\tau_{k}]\cap\R_{+}\times\tilde{\Omega}^{k,n}}+0^{(n)},\\
 & =(H\uI_{(\tau_{k-1},\tau_{k}]\cap\R_{+}\times\Omega^{k,n}}+0^{(n)})\uI_{[0,\alpha]}+0^{(n)},\\
 & =H^{k,n}\I_{[0,\alpha]}\in\cL((X^{k,n})^{\alpha})=\cL((X^{\alpha})^{k,n}).
\end{align*}
Therefore, $H\uI_{[0,\alpha]}+0^{(N^{\alpha})}\in\cL(X^{\alpha})$,
and so
\begin{align*}
(H\uI_{[0,\alpha]}+0^{(N^{\alpha})})\cdot X^{\alpha} & =\sum_{k=1}\sum_{n=1}(H\uI_{[0,\alpha]}+0^{(N^{\alpha})})^{k,n}\cdot(X^{\alpha})^{k,n},\\
 & =\sum_{k=1}\sum_{n=1}(H^{k,n}\I_{[0,\alpha]})\cdot(X^{k,n})^{\alpha},\\
 & =\sum_{k=1}\sum_{n=1}(H^{k,n}\cdot X^{k,n})^{\alpha},\\
 & =(H\cdot X)^{\alpha}.
\end{align*}
The other equality may be proved in a similar fashion to above, yielding
$(H\uI_{[0,\alpha]}+0^{(N)})^{k,n}$$=H^{k,n}\I_{[0,\alpha]}\in\cL(X^{k,n})$,
and $(H\uI_{[0,\alpha]}+0^{(N)})\cdot X=(H\cdot X)^{\alpha}$.
\end{proof}
Mémin's theorem  \cite[Corollary III.4]{Art:Memin:EspaceDeSemimartingale1980}
for a semimartingale $Y$ states that the set of stochastic integrals
$\{H\cdot Y\mid H\in\cL(Y)\}$ is closed in the semimartingale topology.
For details on the semimartingale topology, see \cite{Art:Emery:SemiMart_Topology:1979,Art:Memin:EspaceDeSemimartingale1980}.
This result is used in the proof of the general (sigma-martingale)
version of the fundamental theorem of asset pricing given by Delbaen
and Schachermayer \cite{Art:DelbSchach:FundThmMathFin:1998}, and
also that of Kabanov \cite{Art:Kabanov:OnFTAPKrepsDelbSchach:1997},
who proves a more general intermediate theorem, valid for wealth processes
not necessarily arising from stochastic integration. We extend Mémin's
theorem here in preparation for extending the Delbaen and Schachermayer
FTAP, Theorem \ref{Thm:FTAP} below.

\begin{prop}[Mémin extension]
\label{Prop:Memin}If $X$ is a  piecewise semimartingale, then the
sets of stochastic integrals $\{H\cdot X\mid H\in\cL(X),\; H\cdot X\ge-c\}$
are closed in the semimartingale topology for each $c\in[0,\infty]$. \end{prop}
\begin{proof}
Denote by $\fG_{c}$ the set of stochastic integrals bounded from
below by $-c$, and let $Y$ be in the closure of $\fG_{c}$. Then
there exists a sequence $(H^{i})$ such that $H^{i}\cdot X\in\fG_{c}$,
for all $i\in\N$, and $\lim_{i\to\infty}(H^{i}\cdot X)=Y$ (all limits
in this proof are assumed to be in the semimartingale topology). $Y$
can be dissected as
\begin{align*}
Y & =Y_{0}+\sum_{k=1}^{\infty}\sum_{n=1}^{\infty}Y^{k,n},\qquad Y^{k,n}:=\I_{\Omega^{k,n}}(Y^{\tau_{k}}-Y^{\tau_{k-1}}).
\end{align*}
It is the case that limits in the semimartingale topology may be interchanged
with the operation of stopping a process, a fact proved at the end.
 For example, $\lim_{i\to\infty}(H^{i}\cdot X)^{\tau_{k}}=Y^{\tau_{k}}$,
for each $k\in\N$. Let $H^{k,n,i}$ be the dissection of $H^{i}$
as in (\ref{Eq:H^kn-Def}). We may then deduce that $\lim_{i\to\infty}(H^{k,n,i}\cdot X^{k,n})=Y^{k,n}$,
$\forall k,n\in\N$. The sets
\begin{align*}
\fG^{k,n} & :=\left\{ H^{k,n}\cdot X^{k,n}\mid H^{k,n}\in\cL(X^{k,n})\right\} ,\quad\forall k,n\in\N,
\end{align*}
are closed in the semimartingale topology by Corollary III.4 of \cite{Art:Memin:EspaceDeSemimartingale1980}.
Therefore, there exist processes $\hat{H}^{k,n}\in\fG^{k,n}$ such
that $\hat{H}^{k,n}\cdot X^{k,n}=Y^{k,n}$. Stitching the local pieces
together and choosing $\hat{H}_{0}$ so that $\hat{H}_{0}^{\prime}X_{0}=Y_{0}$
provides the candidate closing integrand,
\begin{align*}
\hat{H} & :=\hat{H}_{0}+\sum_{k=1}^{\infty}\sum_{n=1}^{\infty}(\uI_{(\tau_{k-1},\tau_{k}]\times\R_{+}\cap\Omega^{k,n}}\hat{H}^{k,n})+0^{(N)}.
\end{align*}
Then $\hat{H}\in\cL(X)$ and
\begin{align*}
(\hat{H}\cdot X) & =\hat{H}_{0}X_{0}+\sum_{k=1}^{\infty}\left[(\hat{H}\cdot X)^{\tau_{k}}-(\hat{H}\cdot X)^{\tau_{k-1}}\right],\\
 & =Y_{0}+\sum_{k=1}^{\infty}\sum_{n=1}^{\infty}((\I_{\Omega^{k,n}}\hat{H}^{k,n})\cdot X^{k,n}),\\
 & =Y_{0}+\sum_{k=1}^{\infty}\sum_{n=1}^{\infty}\I_{\Omega^{k,n}}Y^{k,n},\\
 & =Y.
\end{align*}
To show that $Y\ge-c$, semimartingale convergence implies ucp convergence,
which implies that $P(Y_{t}\ge-c,\;0\le s\le t)=1$, $\forall t\ge0$.
Therefore $Y\in\fG_{c}$, so $\fG_{c}$ is closed.

To show that stopping may be interchanged with semimartingale convergence,
we use that $Y^{i}\to Y$ if and only if $(\xi^{i}\cdot Y^{i})_{t}-(\xi^{i}\cdot Y)_{t}\to0$
in probability, for all simple, predictable, bounded sequences of
processes $(\xi^{i})$, $\forall t\ge0$ (see \cite{Art:Emery:SemiMart_Topology:1979,Art:Memin:EspaceDeSemimartingale1980}).
For any stopping time $\alpha$ and any sequence of simple predictable
bounded processes $(\xi^{i})$, $(\xi^{i}\I_{[0,\alpha]})$ is also
a sequence of simple, predictable, bounded processes. Therefore, $Y^{i}\to Y$
implies that $(\xi^{i}\cdot(Y^{i})^{\alpha}-\xi^{i}\cdot Y^{\alpha})_{t}=(\xi^{i}\I_{[0,\alpha]}\cdot Y^{i})_{t}-(\xi^{i}\I_{[0,\alpha]}\cdot Y)_{t}\to0$,
in probability, $\forall t\ge0$.
\end{proof}

\subsection{Martingales\label{Sub:Martingales}}

The notion of martingale and its relatives may also be extended to
the piecewise setting, but due to the reset feature of these processes,
some care is needed. A characterization of martingality that generalizes
usefully is the martingality of stochastic integrals when the integrand
is bounded, simple, and predictable.
\begin{defn}
\label{Def:U_Val_Marts}A\emph{ piecewise martingale} is a  piecewise
semimartingale $X$ such that $H\cdot X$ is a martingale whenever
both $H\in\bS(N)$ and $\left|H\right|_{1}$ is bounded. A\emph{  piecewise
local martingale} is a process $X$ for which there exists a sequence
of increasing stopping times $(\rho_{i})$ such that $\lim_{i\to\infty}\rho_{i}=\infty$,
and $\I_{\{\rho_{i}>0\}}X^{\rho_{i}}$ is a  piecewise martingale,
$\forall i\in\N$. A\emph{ piecewise sigma-martingale} is a  piecewise
semimartingale $X$ such that $H\cdot X$ is a sigma-martingale whenever
$H\in\cL(X)$.
\end{defn}
It is easy to check that the definition of piecewise martingale is
equivalent to the usual definition of martingale, when $X$ is an
$\R^{n}$-valued semimartingale. Hence, the definition of piecewise
local martingale is also equivalent in this case. The equivalence
of the definition of piecewise sigma-martingale follows from \cite[p. 218]{Book:JacodShiryaev:LimitThmsStochProc:2002}.
\begin{rem}
In the definition of  piecewise martingale, requiring $H$ to be bounded
in $|\cdot|_{1}$ rather than in some other local norm is somewhat
arbitrary. All are of course equivalent for $\R^{n}$-valued semimartingales
due to the equivalence of any two norms on a finite-dimensional norm
space. But when the dimension is stochastic and unbounded, then the
definitions depend on the choice of local norm. This distinction disappears
under localization, as Lemma \ref{Lem:Loc_Bdd_Is_Loc_Bdd} shows.
Proposition \ref{Prop:X_Loc_Mart_iff_X^k,n_Loc_Mart} implies that
any choice of local norm in the definition of  piecewise martingale
yields the same class of  piecewise local martingales. See also Corollaries
\ref{Cor:Loc_Mart_Closure_wrt_SI} and \ref{Cor:FTAP_Loc_Bdd}, relating
to this point.\end{rem}
\begin{lem}
\label{Lem:Loc_Bdd_Is_Loc_Bdd}Let $\bK:=\cup_{n\in\N}\cK_{n}$, where
each $\cK_{n}$ is a finite-dimensional normed space. Let $\bK$ be
equipped with the Borel sigma algebra generated by the by the union
of the norm topologies of each $\cK_{n}$. Let $Y$ be a $\bK$-valued
progressive process with associated process $N^{Y}$ satisfying $N^{Y}=n$
whenever $Y\in\cK_{n}$. Suppose that $N^{Y}$ has paths that are
left-continuous with right limits. If $\left\Vert Y\right\Vert _{a}$
is locally bounded for some function $\left\Vert \cdot\right\Vert _{a}:\bK\to[0,\infty)$
such that $\left\Vert \cdot\right\Vert _{a}\restriction_{A}$ is a
norm whenever $A\subset\bK$ is a vector space, then $\left\Vert Y\right\Vert _{b}$
is locally bounded for any $\left\Vert \cdot\right\Vert _{b}:\bK\to[0,\infty)$
such that $\left\Vert \cdot\right\Vert _{b}\restriction_{A}$ is a
norm whenever $A\subset\bK$ is a vector space.\end{lem}
\begin{proof}
Define the stopping times
\begin{align*}
\alpha_{n} & :=\inf\{t\ge0\mid N_{t}>n\},\quad n\in\N.
\end{align*}
Since $N^{Y}$ is left continuous, then $\I_{\{\alpha_{n}>0\}}(N^{Y})^{\alpha_{n}}\le n$.
For each $n\in\N$, $\left\Vert \cdot\right\Vert _{a}\restriction_{\cK_{n}}$
and $\left\Vert \cdot\right\Vert _{b}\restriction_{\cK_{n}}$ are
equivalent. Therefore, for each $n\in\N$, there exists $c_{n}\in(0,\infty)$
such that
\begin{align*}
\left\Vert \I_{\{\alpha_{n}>0\}}Y^{\alpha_{n}}\right\Vert _{b} & \le c_{n}\left\Vert \I_{\{\alpha_{n}>0\}}Y^{\alpha_{n}}\right\Vert _{a}.
\end{align*}
Let $(\beta_{n})$ be a sequence of stopping times such that $\left\Vert \I_{\{\beta_{n}>0\}}Y^{\beta_{n}}\right\Vert _{a}$
is bounded, $\forall n\in\N$, and $\lim_{n\to\infty}\beta_{n}=\infty$.
Then for $\rho_{n}:=\alpha_{n}\wedge\beta_{n}$, $\left\Vert \I_{\{\rho_{n}>0\}}Y^{\rho_{n}}\right\Vert _{b}$
is bounded, $\forall n\in\N$. It remains to show $\alpha_{\infty}:=\lim_{n\to\infty}\alpha_{n}=\infty$.

Since $\bN$ is discrete and $N^{Y}$ has paths with right limits,
then there exists an increasing sequence of random times $(\eta_{n})$
such that $\lim_{n\to\infty}\eta_{n}=\alpha_{\infty}$, and $N_{\eta_{n}}^{Y}=N_{\alpha_{n}^{+}}^{Y}\ge n$
on $\{\alpha_{n}<\infty\}$. Therefore, $\lim_{n\to\infty}N_{\eta_{n}}^{Y}$$=\lim_{n\to\infty}N_{\alpha_{n}^{+}}^{Y}=\infty$
on $\{\alpha_{\infty}<\infty\}$. This contradicts the left continuity
of the paths of $N^{Y}$, thus $P(\alpha_{\infty}<\infty)=0$.
\end{proof}
Lemma \ref{Lem:Loc_Bdd_Is_Loc_Bdd} invites an unambiguous extension
of the notion of a \emph{locally bounded} process taking values in
a finite-dimensional norm space.
\begin{defn}
A process $Y$ meeting the conditions of Lemma \ref{Lem:Loc_Bdd_Is_Loc_Bdd}
is a \emph{locally bounded }process.
\end{defn}
We will next formulate some characterizations of Definition \ref{Def:U_Val_Marts}
that will be useful in the sequel. First we need the following lemma,
for use in conjunction with dissection.
\begin{lem}
\label{Lem:Partitioned-Local-Marts-Are-Local-Marts}If $\eta$ is
a stopping time, $(C_{j})_{j\in\N}$ is an $\cF_{\eta}$-measurable
partition of $\Omega$, and $Y$ is an $\R^{n}$-valued semimartingale,
equal to $0^{(n)}\in\R^{n}$ on $[0,\eta]$, then
\begin{enumerate}
\item If $Y\I_{C_{j}}$ is a martingale for each $j\in\N$, and $Y$ is
an $L_{1}$-process, then $Y$ is a martingale.
\item If $Y\I_{C_{j}}$ is a supermartingale for each $j\in\N$, and $Y$
is an $L_{1}$-process, then $Y$ is a supermartingale.
\item If $Y\I_{C_{j}}$ is a local martingale for each $j\in\N$, then $Y$
is a local martingale.
\item If $Y\I_{C_{j}}$ is a sigma-martingale for each $j\in\N$, then $Y$
is a sigma-martingale.
\end{enumerate}
\end{lem}
\begin{proof}
Starting with the martingale case and using dominated convergence,
we have for $0\le s\le t<\infty$,
\begin{align*}
E[Y_{t}\mid\cF_{s}] & =\sum_{j=1}^{\infty}E[\I_{C_{j}}Y_{t}\mid\cF_{s}]=\sum_{j=1}^{\infty}\I_{C_{j}}Y_{s}=Y_{s}.
\end{align*}
Replacing the second equality with ``$\le$'' proves the supermartingale
case.

For the local martingale case, let $(\rho_{i}^{j})_{i\in\N}$ be a
fundamental sequence for $\I_{C_{j}}Y$ for each $j\in\N$. Define
$\rho_{i}:=\eta\vee\sum_{j=1}^{i}\I_{C_{j}}\rho_{i}^{j}$, $\forall i\in\N$.
Then the $\rho_{i}$ are stopping times because $\{\rho_{i}\le t\}=\{\rho_{i}\le t,\;\eta\le t\}=\cup_{j=1}^{\infty}C_{j}\cap\{\rho_{i}\le t,\;\eta\le t\}=\cup_{j=1}^{i}C_{j}\cap\{\rho_{i}^{j}\le t,\;\eta\le t\}\cup_{j>i}C_{j}\cap\{\eta\le t\}\in\cF_{t}$,
since $C_{j}\cap\{\eta\le t\}\in\cF_{t}$, and the $\rho_{i}^{j}$
are stopping times. Making use of the fact that $Y$ is $0^{(n)}$
on $[0,\eta]$,
\begin{align*}
E[Y_{t}^{\rho_{i}}\mid\cF_{s}] & =\sum_{j=1}^{i}E[\I_{C_{j}}Y{}_{t}^{\rho_{i}^{j}}\mid\cF_{s}]+0^{(n)}=\sum_{j=1}^{i}\I_{C_{j}}Y{}_{s}^{\rho_{i}^{j}}+Y_{s}\sum_{j=i+1}^{\infty}\I_{C_{j}}=Y{}_{s}^{\rho_{i}}.
\end{align*}
If $Y\I_{C_{j}}$ is a sigma-martingale, then $Y\I_{C_{j}}=H^{j}\cdot M^{j}$
for some martingale $M^{j}$ and some $H\in\cL(M^{j})$. Since $Y$
is zero on $[0,\eta]$, then we may take $M^{j}$ and $H^{j}$ to
be zero on this set also. By the previous property, $M:=\sum_{j=1}^{\infty}\I_{C_{j}}M^{j}$
is a local martingale. The process $H:=\sum_{j=1}^{\infty}H^{j}\I_{C_{j}}$
satisfies $H\in\cL(M)$, and $H\cdot M=Y$. Therefore, $Y$ is a sigma-martingale
by \cite[p. 238]{Book:Protter:SDE:2005}.
\end{proof}
If $X$ is a  piecewise martingale, then it is an easy consequence
of the definition that $X^{k,n}$ is a martingale, for each $k,n\in\N$.
However, due to the reset feature of piecewise processes, the converse
is false, even if additionally $|X|_{1}$ is bounded.
\begin{example}[Bounded piecewise strict local martingale]
\label{Ex:Bounded_Piecewise_Strict_Local_Mart}Let $Y$ be any $\R^{n}$-valued
continuous-path \emph{strict} local martingale (that is, a local martingale
that is not a martingale) with $Y_{0}=0^{(n)}$. Let $\tau_{0}:=0$,
and $\tau_{k}:=\inf\{t>\tau_{k-1}\mid|Y_{t}-Y_{\tau_{k-1}}|_{1}=1\}$.
We have that $\tau_{k}\nearrow\infty$ since $Y$ has finite quadratic
variation. Furthermore, $|Y^{\tau_{k}}|_{1}\le k$, so $Y^{\tau_{k}}$
is a bounded local martingale, hence a martingale, for each $k\in\N$.

Define the piecewise $\R^{n}$-valued process $X$ via
\begin{align*}
X & :=\sum_{k=1}^{\infty}\I_{(\tau_{k-1},\tau_{k}]}X^{k,n},\\
X^{k,n}: & =Y^{\tau_{k}}-Y^{\tau_{k-1}},\quad\forall k\in\N.
\end{align*}
This definition implies that $|X|_{1}\le1$, and that $(\tau_{k})$
is a reset sequence for $X$. The processes $X^{k,m}$ are bounded
local martingales, hence martingales, $\forall k,m\in\N$, making
$X$ a piecewise local martingale. Since $Y$ is a strict local martingale,
then there exists $S\in\bS(n)$ with $\left|S\right|_{1}$ bounded,
such that $S\cdot Y$ is not a martingale. But $\cL(X)=\cL(Y)$, and
$H\cdot X=H\cdot Y$, $\forall H\in\cL(X)$. Hence, $S\cdot X$ is
not a martingale, and so $X$ is not a piecewise martingale.
\end{example}
The notions of local martingale and sigma-martingale hold globally
if and only if they hold locally. This idea is made precise in the
following characterizations of the piecewise notions via the properties
holding on each piece.
\begin{prop}
\label{Prop:X_Loc_Mart_iff_X^k,n_Loc_Mart}A  piecewise semimartingale
$X$ is a  piecewise local martingale if and only if for all reset
sequences $(\tau_{k})$, $X^{k,n}$ is a local martingale for all
$k,n\in\N$, if and only if for some reset sequence $(\tau_{k})$,
$X^{k,n}$ is a local martingale for all $k,n\in\N$.\end{prop}
\begin{proof}
To show that the first condition implies the middle condition, suppose
that $X$ is a local martingale and $(\rho_{i})$ is a fundamental
sequence for $X$. Then for $H\in\bS(N^{\rho_{i}})$, $\left|H\right|_{1}$
bounded, $H\cdot(\I_{\{\rho_{i}>0\}}X^{\rho_{i}})$ is a martingale.
Let $G$ be $\R^{n}$-valued, simple, predictable, and $\left|G\right|_{1}$
(or equivalently $\left|G\right|_{p}$, for $p\in[1,\infty]$) be
bounded. Define $H:=G\uI_{(\tau_{k-1},\tau_{k}]\cap\R_{+}\times\Omega^{k,n}}+0^{(N^{\rho_{i}})}$,
resulting in $H\in\bS(N^{\rho_{i}})$, $\left|H\right|_{1}$ bounded,
and $H\cdot(\I_{\{\rho_{i}>0\}}X^{\rho_{i}})=G\cdot(\I_{\{\rho_{i}>0\}}X^{\rho_{i}})^{k,n}$
is a martingale. Thus $(\I_{\{\rho_{i}>0\}}X^{\rho_{i}})^{k,n}$ is
a martingale, and since $(X^{\rho_{i}})^{k,n}=(X^{k,n})^{\rho_{i}}$
by Proposition \ref{Prop:StochIntChars:Stopping}, then $(\rho_{i})$
is a fundamental sequence for $X^{k,n}$, which is therefore a local
martingale, $\forall k,n\in\N$.

The middle condition obviously implies the last condition. To show
that the last implies the first, fix some reset sequence $(\tau_{k})$
and suppose that $X^{k,n}$ is a local martingale, $\forall k,n\in\N$.
Let $(\rho_{i}^{k,n})_{i\in\N}$ be a fundamental sequence for $X^{k,n}$,
and note that since $X_{0}^{k,n}=0^{(n)}$, we have $\I_{\{\alpha>0\}}(X^{k,n})^{\alpha}$$=(X^{k,n})^{\alpha}$,
for any stopping time $\alpha$, and $\forall k,n\in\N$. Define $\hat{\rho}_{i}^{k}:=\tau_{k-1}\vee(\sum_{n=1}^{i}\I_{\Omega^{k,n}}\rho_{i}^{k,n})$,
$\forall i,k\in\N$. Then $\hat{\rho}_{i}^{k}$ are stopping times,
because $\{\hat{\rho}_{i}^{k}\le t\}=\{\hat{\rho}_{i}^{k}\le t,\;\tau_{k-1}\le t\}=(\cup_{n=1}^{i}\Omega^{k,n}\cap\{\rho_{i}^{k,n}\le t,\tau_{k-1}\le t\})\cup_{n>i}\Omega^{k,n}\cap\{\tau_{k-1}\le t\}\in\cF_{t}$,
since $\Omega^{k,n}\cap\{\tau_{k-1}\le t\}\in\cF_{t}$, and $\rho_{i}^{k,n}$
are stopping times. Let $H$ satisfy $H\in\bS(N)$, $H_{0}=0^{(N_{0})}$,
and $\left|H\right|_{1}$ is bounded. Then $\left(\sum_{n=1}^{\infty}\I_{\Omega^{k,n}}H^{k,n}\cdot X^{k,n}\right)^{\hat{\rho}_{i}^{k}}=\sum_{n=1}^{i}\I_{\Omega^{k,n}}(H^{k,n}\cdot X^{k,n})^{\rho_{i}^{k,n}}$
is a martingale, since it is a finite sum of martingales. Then for
each fixed $k$, $\alpha_{i}^{k}:=\bigwedge_{m=1}^{k}\hat{\rho}_{i}^{m}$
is a fundamental sequence in $i\in\N$ for $(H\cdot X)^{\tau_{k}}=\sum_{m=1}^{k}\sum_{n=1}^{\infty}\I_{\Omega^{m,n}}(H^{m,n}\cdot X^{m,n})$.
 To get a fundamental sequence for $X$, for each $k$ let $i=i(k)$
be large enough such that $P(\alpha_{i(k)}^{k}<\tau_{k}\wedge k)<2^{-k}$.
Then $\lim_{k\to\infty}\alpha_{i(k)}^{k}=\infty$ a.s., and $\alpha_{i(k)}^{k}$
reduces $(H\cdot X)^{\tau_{k}}$ for each $k$. Therefore each $\hat{\alpha}_{p}:=\max(\alpha_{i(1)}^{1}\wedge\tau_{1},\ldots,\alpha_{i(p)}^{p}\wedge\tau_{p})$
reduces $H\cdot X$, $\forall p\in\N$, and $\lim_{p\to\infty}\hat{\alpha}_{p}=\infty$
a.s. This sequence does not depend on $H$, but we assumed $H_{0}=0^{(N_{0})}$.
To allow for $H$ with nonzero $H_{0}$, define $\beta_{p}:=\hat{\alpha}^{p}\I_{\{|X_{0}|_{1}\le p\}}$.
Then $\lim_{p\to\infty}\beta_{p}=\infty$ a.s., $H_{0}^{\prime}X_{0}\I_{\{\beta_{p}>0\}}\in L_{1}$,
and $\I_{\{\beta_{p}>0\}}(H\cdot X)^{\beta_{p}}$ is a martingale.
So $(\beta_{p})$ is fundamental for $X$, and therefore $X$ is a
piecewise local martingale.\end{proof}
\begin{prop}
\label{Prop:X_Sig_Mart_iff_X^k,n_Sig_Mart}A  piecewise semimartingale
$X$ is a  piecewise sigma-martingale if and only if for all reset
sequences $(\tau_{k})$, $X^{k,n}$ is a sigma-martingale for all
$k,n\in\N$, if and only if for some reset sequence $(\tau_{k})$,
$X^{k,n}$ is a sigma-martingale for all $k,n\in\N$.\end{prop}
\begin{proof}
To show that the first condition implies the middle condition, suppose
that $X$ is a  piecewise sigma-martingale. Then for arbitrary reset
sequence $(\tau_{k})$, define the simple processes
\begin{align}
G^{k,n,i} & :=(\underbrace{0,\ldots,0}_{i-1},1,\underbrace{0,\ldots,0}_{n-i})\uI_{(\tau_{k-1},\tau_{k}]\cap\R_{+}\times\Omega^{k,n}}+0^{(N)}\in\cL_{0}(X),\quad k,n,i\in\N.\label{Eq:G_Integrand_Picks_Out_Component}
\end{align}
Then $G^{k,n,i}\cdot X$$=X_{i}^{k,n}$, which must be an $\R$-valued
sigma-martingale by the definition of $X$ being a piecewise sigma-martingale.
Therefore $X^{k,n}=(X_{1}^{k,n},\ldots,X_{n}^{k,n})$ is an $\R^{n}$-valued
sigma-martingale.

The middle condition obviously implies the last condition. To show
that the last implies the first, fix some reset sequence $(\tau_{k})$,
such that $X^{k,n}$ is an $\R^{n}$-valued sigma-martingale for each
$k,n\in\N$. If $H\in\cL(X)$, then $H^{k,n}\cdot X^{k,n}$ exists,
and is an $\R$-valued sigma-martingale, since sigma-martingales are
closed under stochastic integration. Then $\sum_{n=1}^{\infty}H^{k,n}\cdot X^{k,n}=\sum_{n=1}^{\infty}\I_{\Omega^{k,n}}H^{k,n}\cdot X^{k,n}$
is a sigma-martingale by Lemma \ref{Lem:Partitioned-Local-Marts-Are-Local-Marts},
and $(H\cdot X)^{\tau_{k}}=\sum_{j=1}^{k}\sum_{n=1}^{\infty}H^{j,n}\cdot X^{j,n}$
is a sigma-martingale, because sigma-martingales form a vector space.
Then $H\cdot X$ is a sigma-martingale because any process that is
locally a sigma-martingale is a sigma-martingale \cite[p. 238]{Book:Protter:SDE:2005}.
\end{proof}
The following theorem is an extension of the Ansel and Stricker theorem
\cite{Art:AnselStricker:Thm:1994} to piecewise martingales. The theorem
provides a necessary and sufficient characterization of when the local
martingality property is conserved with respect to stochastic integration.
\begin{thm}[Ansel and Stricker extension]
\label{Thm:Ansel_Stricker} Let $X$ be a  piecewise local martingale,
and $H\in\cL(X)$. Then $H\cdot X$ is a local martingale if and only
if there is an increasing sequence of stopping times $\alpha_{j}\nearrow\infty$
and a sequence $(\vartheta_{j})$ of $L_{1}$, $(-\infty,0]$-valued
random variables such that $(H^{\prime}\Delta X)^{\alpha_{j}}\ge\vartheta_{j}$.\end{thm}
\begin{proof}
The proof of necessity relies only on the fact that $H\cdot X$ is
a local martingale, not that it is a stochastic integral. We reproduce
the proof given in Theorem 7.3.7 of \cite{Book:DelbSchach:ArbBook:2006}.
$H\cdot X$ is a local martingale implies that it locally has finite
$\cH^{1}$ norm, where $\left\Vert M\right\Vert _{\cH^{1}}:=E([M,M]^{1/2})$
for $\R$-valued local martingale $M$. Hence the Burkholder-Davis-Gundy
inequalities \cite[Theorem IV.48]{Book:Protter:SDE:2005}) imply that
there exists an increasing sequence of stopping times $\alpha_{j}\nearrow\infty$
such that $\sup_{t\le\alpha_{j}}|(H\cdot X)_{t}|\in L_{1}$. Therefore,
$|H^{\prime}\Delta X|^{\alpha_{j}}\le$$2\sup_{t\le\alpha_{j}}|(H\cdot X)_{t}|\in L_{1}$.

For sufficiency, suppose that there exists a sequence of increasing
stopping times $\alpha_{j}\nearrow\infty$ and a sequence of $L_{1}$,
nonpositive random variables $(\vartheta_{j})$ such that $(H^{\prime}\Delta X)^{\alpha_{j}}\ge\vartheta_{j}$.
Then by dissection, $((H^{k,n})^{\prime}\Delta X^{k,n})^{\alpha_{j}}\ge\vartheta_{j}$.
By Proposition \ref{Prop:X_Loc_Mart_iff_X^k,n_Loc_Mart}, $X^{k,n}$
is a local martingale, so by the Ansel and Stricker theorem in $\R^{n}$,
$H^{k,n}\cdot X^{k,n}$ is a local martingale, $\forall k,n\in\N$.
Hence Lemma \ref{Lem:Partitioned-Local-Marts-Are-Local-Marts} implies
that $\sum_{n=1}^{\infty}H^{k,n}\cdot X^{k,n}$ is a local martingale,
and therefore $(H\cdot X)^{\tau_{m}}=$$\sum_{k=1}^{m}\sum_{n=1}^{\infty}H^{k,n}\cdot X^{k,n}$
is a local martingale, being a finite sum of local martingales. Therefore
$H\cdot X$ is locally a local martingale, so is itself a local martingale.
\end{proof}
In $\R^{n}$-stochastic analysis the set of local martingales has
several other useful closure properties with respect to stochastic
integration. Below are generalizations of a few of these properties
to  piecewise local martingales.
\begin{cor}
\label{Cor:Loc_Mart_Closure_wrt_SI}Let $X$ be a piecewise local
martingale.
\begin{enumerate}
\item If $H\in\cL(X)$, and $H$ is locally bounded, then $H\cdot X$ is
a local martingale.
\item If $X$ is left-continuous, and $H\in\cL(X)$, then $H\cdot X$ is
a local martingale.
\end{enumerate}
\end{cor}
\begin{proof}

The processes $X^{k,n}$ are local martingales for all $k,n\in\N$,
by Proposition \ref{Prop:X_Loc_Mart_iff_X^k,n_Loc_Mart}. For the
first case, $H$ locally bounded implies that $H^{k,n}$ is locally
bounded, $\forall k,n\in\N$. $\R^{n}$-stochastic integration preserves
local martingality with respect to locally bounded integrands (\cite[Theorem IV.29]{Book:Protter:SDE:2005}),
therefore $H^{k,n}\cdot X^{k,n}$ is a local martingale, $\forall k,n\in\N$.
Using the same argument at the end of the proof of the Ansel and Stricker
theorem, then $H\cdot X$ is a local martingale.

In the second case, $X$ is left continuous, so $\Delta X=X-X^{-}=0$.
Hence, Theorem \ref{Thm:Ansel_Stricker} implies that $H\cdot X$
is a local martingale.
\end{proof}

\section{\label{Sec:MarketModelsStochNumAssets}Arbitrage in piecewise semimartingale
market models}

\subsection{Market models with a stochastic number of assets\label{Sub:Market-models-with-Stoch-Num-Assets}}

In this section we will specify the market model for an investable
universe having a finite, but unbounded, stochastic number of assets
available for investment. The process $X$ is a  piecewise semimartingale
modeling the prices of the $N=\dim X$ assets. At each $\tau_{k}^{+}$
the market prices may reconfigure in an arbitrary way, potentially
adding or removing assets.

There is a money market account $B$, which for convenience we give
an interest rate of zero so that $B=1$. The process $V^{v,H}$ is
the total wealth of an investor starting with initial wealth $V_{0}=v$,
and investing by holding $H\in\cL_{0}(X)$ shares of the risky assets.
All wealth processes are assumed to be \emph{self-financing,} meaning
that there exists a process $H\in\cL_{0}(X)$ and initial wealth $v\in\R$
such that
\begin{align*}
V_{t}^{v,H} & =v+(H\cdot X)_{t},\quad\forall t\ge0.
\end{align*}

\begin{rem}
It is important to be clear about what this self-financing condition
implies in the model. Since $H\cdot X$ is right-continuous, it is
unaffected by any discontinuities in $X$ at $\tau_{k}^{+}$. Therefore,
self financing portfolios will not be affected by these jumps. This
is useful for modeling certain types of events normally excluded from
equity market models: the entry of new companies,  the merging of
several companies into one, and the breakup or spin off of a company
into several companies.

These events may affect portfolio values \emph{upon their announcement},
but leave them unaffected at the point in time of their implementation.
Any surprise in the announcement of such events can be manifested
through a left discontinuity in $X$, which is passed on to $V$.
Furthermore, there need not be any gap between announcement and implementation,
since the paths of $X$ may have both a right and left discontinuity
at a given point in time. An illustrative example is when a company
goes bankrupt via a jump to $0$ in its stock price. This should occur
via a left discontinuity, since this event should affect portfolio
values through $H\cdot X$. A right discontinuity may also occur at
this time, as the market transitions from $n$ to $n-1$ assets, removing
the bankrupt company as an option for investment.
\end{rem}
We assume the standard notion of admissibility, that a trading strategy
must have a limited credit line. That is, losses must be uniformly
bounded from below.
\begin{defn}
\emph{\label{Def:AdmTradStrat-1}}A  process $H$ is \emph{admissible
}for piecewise semimartingale $X$ if both of the following hold:
\begin{enumerate}
\item \label{Def:AdmTradStrat:Item:Integrable-1}$H\in\cL_{0}(X)$;
\item \label{Def:AdmTradStrat:Item:BB-1}There exists a constant $c$ such
that a.s.
\begin{align}
(H\cdot X)_{t} & \ge-c,\quad\forall t\ge0.\label{Eq:AdmStratsPartitionBB-1}
\end{align}

\end{enumerate}
The class of \emph{nonnegative wealth processes} is denoted by
\begin{align*}
\cV & :=\cV(X):=\{V^{v,H}:=v+H\cdot X\mid v\in\R_{+},\; H\in\cL_{0}(X),\; v+H\cdot X\ge0\}.
\end{align*}

\end{defn}

\subsection{Fundamental theorems of asset pricing\label{Sub:FTAPs}}

Characterizing the presence or absence of various arbitrage-like notions
in a market model is important for both checking the viability of
a model for the purposes of optimization and realism, and conversely
for discovering portfolios that may be desirable to implement in practice.
In this section we study the existence of arbitrage of the first kind
and free lunch with vanishing risk, giving FTAPs for each.

The presence of arbitrage of the first kind, studied recently by Kardaras
in \cite{Art:Kardaras:FiniteAddProb&FTAP:2010,Art:Kardaras:MarkViabAbsArb1stKind:2012},
may be a sufficiently strong pathology to rule out a market model
for practical use. In other words, its absence, $\mbox{NA}_{1}$,
is often viewed as a minimal condition for market viability. The notion
of arbitrage of the first kind has previously appeared in the literature
under several different names and equivalent formulations. The name\emph{
cheap thrill }was used in \cite{Art:LoewWillard:LocMartsFreeSnacksCheapThrills:2000},
whereas the property of the set $\{V\in\cV(X)\mid V_{0}=1\}$ being
bounded in probability, previously called \emph{BK }in \cite{Art:Kabanov:OnFTAPKrepsDelbSchach:1997}
and \emph{no unbounded profit with bounded risk }(\emph{NUPBR) }in
\cite{Art:KaratzKard:Numeraire:2007}, was shown in \cite[Proposition 1.2]{Art:Kardaras:FiniteAddProb&FTAP:2010}
to be equivalent to $\mbox{NA}_{1}$.

The condition NFLVR is stronger than $\mbox{NA}_{1}$. It was studied
by Delbaen and Schachermayer in \cite{Art:DelbSchach:FundThmMathFin:1994,Art:DelbSchach:FundThmMathFin:1998,Book:DelbSchach:ArbBook:2006},
and it rules out approximate arbitrage in a sense recalled below.
In certain market models, such as those admitting arbitrage in stochastic
portfolio theory \cite{Book:Fernholz:SPT:2002,Art:Karatzas&Fernholz:SPTReview:2009}
and the benchmark approach \cite{Book:PlatenHeath:BenchmarkApproach:2006},
the flexibility of violating NFLVR while upholding $\mbox{NA}_{1}$
is essential.
\begin{defn}
An \emph{arbitrage of the first kind} for $X$ on horizon $\alpha$,
a finite stopping time, is an $\cF_{\alpha}$-measurable random variable
$\psi$ such that $P[\psi\ge0]=1$, $P[\psi>0]>0$ and, for each $v>0$,
there exists $V^{v,H}\in\cV(X)$, where $H$ may depend on $v$, satisfying
$v+(H\cdot X)_{\alpha}\ge\psi$. If there exists no arbitrage of the
first kind, then $\mbox{\emph{NA}}_{1}$ holds.
\end{defn}
Note that in contrast to NFLVR, Definition \ref{Def:NFLVR}, $\mbox{NA}_{1}$
and its corresponding FTAP, Theorem \ref{Thm:NA1_ELMD}, may be formulated
on a horizon that is an \emph{unbounded} stopping time. In this sense,
it is a more general result than the NFLVR FTAP.
\begin{defn}
\label{Def:ELMD}An \emph{equivalent local martingale deflator }(ELMD)
for $\cV(X)$ is a strictly positive $\R$-valued local martingale
$Z$, such that $Z_{0}=1$, and for each $V\in\cV(X)$, $ZV$ is a
nonnegative local martingale.
\end{defn}
An ELMD is identical to the notion of \emph{strict martingale density},
as in \cite{Art:Schweizer-On_The_Minimal_Mart_Meas:1995}. For an
FTAP relating ELMDs to finitely additive, locally equivalent probability
measures for $\R$-valued $X$, see \cite{Art:Kardaras:FiniteAddProb&FTAP:2010}.

When the price process is an $\R$-valued semimartingale, Kardaras
proved in Theorem 1.1 of \cite{Art:Kardaras:MarkViabAbsArb1stKind:2012}
that $\mbox{NA}_{1}$ is equivalent to the existence of an ELMD. We
assume here, as he does in \cite{Art:Kardaras:MarkViabAbsArb1stKind:2012},
that the result extends to $\R^{n}$-valued semimartingales, and use
this to prove that it holds for piecewise semimartingales as well.
In performing the extension, it is useful to recruit the $n$-dimensional
``market slices'' running on stochastic time intervals $(\tau_{k-1},\tau_{k}]$.
These slices can be taken as markets in and of themselves, with price
processes $X^{k,n}$.
\begin{thm}
\label{Thm:NA1_ELMD}Let $\alpha$ be a finite stopping time. $\mbox{NA}_{1}$
holds for $X$ on horizon $\alpha$ if and only if it holds for each
$X^{k,n}$, $k,n\in\N$, on horizon $\alpha$, if and only if there
exists an ELMD for $\cV(X^{\alpha})$. \end{thm}
\begin{proof}
The strategy of the proof is to prove the implications $\mbox{NA}_{1}\mbox{ for }X\imply\mbox{NA}_{1}$
for each $X^{k,n}\imply$ ELMD for $\cV(X^{\alpha})\imply\mbox{NA}_{1}$
for $X$.

$(\mbox{NA}_{1}$ for $X\imply\mbox{NA}_{1}$ for each $(X^{k,n})$)
We prove the contrapositive. Suppose there exists an arbitrage of
the first kind $\psi$ with respect to $X^{k,n}$. Let $H^{k,n,v}\in\cL_{0}(X^{k,n})$
and satisfy $v+(H^{k,n,v}\cdot X^{k,n})_{\alpha}\ge\psi$. Define
$H^{v}:=H^{k,n,v}\uI_{(\tau_{k-1},\tau_{k}]\cap\R_{+}\times\Omega^{k,n}}+0^{(N)}$,
which satisfies $H^{v}\in\cL_{0}(X)$, and $H^{v}\cdot X=H^{k,n,v}\cdot X^{k,n}$.
Therefore, for each $v>0$, there exists $H^{v}\in\cL_{0}(X)$ such
that $H^{v}\cdot X\ge-v$ and $v+(H^{v}\cdot X)_{\alpha}\ge\psi$,
and so $\psi$ is an arbitrage of the first kind with respect to $X$.

($\mbox{NA}_{1}$ for each $(X^{k,n})\imply$ ELMD for $\cV(X^{\alpha})$)
By Theorem 1.1 of \cite{Art:Kardaras:MarkViabAbsArb1stKind:2012},
for each $k,n\in\N$, there exists $Z^{k,n}$, an ELMD for $\cV((X^{k,n})^{\alpha})$.
Without loss of generality (for example, substitute $Z^{k,n}$ with
$\frac{(Z^{k,n})^{\tau_{k}}}{(Z^{k,n})^{\tau_{k-1}}}$), we may take
$Z^{k,n}=1$ on $[0,\tau_{k-1}]$, since $X^{k,n}=0$ here.  Define
$Z^{k}:=\I_{\{\tau_{k-1}=\infty\}}+\sum_{n=1}^{\infty}\I_{\Omega^{k,n}}(Z^{k,n})^{\tau_{k}}$,
$Z=\prod_{k=1}^{\infty}Z^{k}$, and for $X^{\alpha}$-admissible trading
strategy $H$, and $v\in\R$, define $Y:=v+H\cdot X^{\alpha}$, $Y^{k}:=Y^{\tau_{k}}-Y^{\tau_{k-1}}$,
$Y^{k,n}:=Y^{k}\I_{\Omega^{k,n}}=H^{k,n}\cdot(X^{k,n})^{\alpha}$.
 Although $H$ being admissible for $X^{\alpha}$ does not in general
imply that $H^{k,n}$ is admissible for $(X^{k,n})^{\alpha}$, this
is easily remedied by further dissection of $H^{k,n}$ into pieces
using the $\cF_{\tau_{k-1}}$-measurable partition:
\begin{align*}
 & \left\{ \{\tau_{k-1}=\infty\},\quad C_{j}:=\{\tau_{k-1}<\infty\}\cap\{j\le(H\cdot X^{\alpha})_{\tau_{k-1}}<j+1\},\quad j\in\bZ\right\} .
\end{align*}
Then $\I_{C_{j}}(H^{k,n}\cdot(X^{k,n})^{\alpha})$ must be uniformly
bounded from below, since $H\cdot X^{\alpha}$ is, and $\I_{C_{j}}(H\cdot X^{\alpha})_{\tau_{k-1}}<j+1$.
By definition $H^{k,n}=0^{(n)}$ on $[0,\tau_{k-1}]$, thus $\I_{C_{j}}H^{k,n}$
is predictable, and so is $(X^{k,n})^{\alpha}$-admissible. Therefore,
$Z^{k,n}\bigl(\I_{C_{j}}(H^{k,n})\cdot(X^{k,n})^{\alpha}\bigr)=\I_{C_{j}}Z^{k,n}Y^{k,n}$
is a local martingale for each $k,n\in\N$, $j\in\bZ$. Lemma \ref{Lem:Partitioned-Local-Marts-Are-Local-Marts}
implies that $Z^{k,n}Y^{k,n}=\sum_{j=-\infty}^{\infty}\I_{C_{j}}Z^{k,n}Y^{k,n}$
is a local martingale, and furthermore that $Z^{k}Y^{k}=\sum_{n=1}^{\infty}\I_{\Omega^{k,n}}Z^{k,n}Y^{k,n}$
is also a local martingale, noting that both processes are zero on
$[0,\tau_{k-1}]$.

We prove by induction that $(ZY)^{\tau_{k}}$ is a local martingale,
$\forall k\in\N$. First, $(ZY)^{\tau_{1}}=Z^{1}(v+Y^{1})$ is a local
martingale by the above. Assume that $(ZY)^{\tau_{k-1}}$ and $Z^{\tau_{k-1}}$
are local martingales, and choose a fundamental sequence $(\rho_{j})$
that is a common reducing sequence for $Z^{\tau_{k-1}}$, $(ZY)^{\tau_{k-1}}$,
$Z^{k}$, and $Z^{k}Y^{k}$, which can always be done by taking the
minimum at each index over a reducing sequence for each. Making repeated
use of $(Z^{k}Y^{k})_{t}=(Z^{k}Y^{k})_{t\vee\tau_{k-1}}$, we have
for $0\le s\le t<\infty$ that $E[(ZY)_{t}^{\tau_{k}\wedge\rho_{j}}\mid\cF_{s}]$
\begin{align*}
 & =E\left[\Bigl(\prod_{i=1}^{k}Z^{i}\sum_{m=1}^{k}Y^{m}{}_{t}^{\rho_{j}}\Bigr)\mid\cF_{s}\right],\\
 & =E\left[Z_{t}^{\tau_{k-1}\wedge\rho_{j}}\left(E\left[(Z^{k}Y^{k})_{t}^{\rho_{j}}\mid\cF_{s\vee\tau_{k-1}}\right]+Y_{t}^{\tau_{k-1}\wedge\rho_{j}}E\left[(Z^{k})_{t}^{\rho_{j}}\mid\cF_{s\vee\tau_{k-1}}\right]\right)\mid\cF_{s}\right],\\
 & =E\left[Z_{t}^{\tau_{k-1}\wedge\rho_{j}}\left(E\left[(Z^{k}Y^{k})_{t\vee\tau_{k-1}}^{\rho_{j}}\mid\cF_{s\vee\tau_{k-1}}\right]+Y_{t}^{\tau_{k-1}\wedge\rho_{j}}E\left[(Z^{k})_{t\vee\tau_{k-1}}^{\rho_{j}}\mid\cF_{s\vee\tau_{k-1}}\right]\right)\mid\cF_{s}\right],\\
 & =E\left[Z_{t}^{\tau_{k-1}\wedge\rho_{j}}(Z^{k})_{s\vee\tau_{k-1}}^{\rho_{j}}\left((Y^{k})_{s\vee\tau_{k-1}}^{\rho_{j}}+Y_{t}^{\tau_{k-1}\wedge\rho_{j}}\right)\mid\cF_{s}\right],\\
 & =(Z^{k})_{s}^{\rho_{j}}\left[(Y^{k})_{s}^{\rho_{j}}E\left(Z_{t}^{\tau_{k-1}\wedge\rho_{j}}\mid\cF_{s}\right)+E\left((ZY)_{t}^{\tau_{k-1}\wedge\rho_{j}}\mid\cF_{s}\right)\right],\\
 & =(Z^{k})_{s}^{\rho_{j}}[(Y^{k})_{s}^{\rho_{j}}Z_{s}^{\tau_{k-1}\wedge\rho_{j}}+(ZY)_{s}^{\tau_{k-1}\wedge\rho_{j}}],\\
 & =(ZY)_{s}^{\tau_{k}\wedge\rho_{j}}.
\end{align*}
Therefore, $(ZY)^{\tau_{k}}$ is a local martingale, and by choosing
$Y=1$ ($v=1,H=0^{(N^{\alpha})}$), $Z^{\tau_{k}}$ can be seen to
be a local martingale as well, completing the induction.

Since any process that is locally a local martingale is a local martingale,
then $ZY$ and $Z$ are local martingales. $Z$ is strictly positive,
since for $P$-almost every $\omega$, it is the product of finitely
many strictly positive terms. Thus, $Z$ is an ELMD for $\cV(X^{\alpha})$.

(There exists $Z$ ELMD for $\cV(X^{\alpha})\imply\mbox{NA}_{1}$
for $X$.) Suppose that $\psi$ is a nonnegative random variable,
and that there exists a sequences $(H^{j})$, such that $H^{j}\in\cL_{0}(X)$,
$\forall j\in\N$, and there exists a sequence of nonnegative numbers
$(v_{j})\searrow0$, such that $v_{j}+(H^{j}\cdot X)_{\alpha}\ge\psi$,
$\forall j\in\N$. Then using the fact that nonnegative local martingales
are supermartingales, we obtain
\begin{align*}
E[Z_{\alpha}\psi] & \le E[Z_{\alpha}(v_{j}+(H^{j}\cdot X)_{\alpha})]\le Z_{0}v_{j}=v_{j},\quad\forall j\in\N.
\end{align*}
Since $\lim_{j\to\infty}v_{j}=0$, this implies that $E[Z_{\alpha}\psi]\le0$.
Hence, $\psi=0$ and is not an arbitrage of the first kind. Thus $\mbox{NA}_{1}$
holds for $X$.
\end{proof}
Theorem \ref{Thm:NA1_ELMD} can be described as holding globally if
and only if it holds locally. This makes it very convenient and easy
to verify in practice compared to the NFLVR notion, which can hold
locally without holding globally.
\begin{defn}
\label{Def:NFLVR}For piecewise semimartingale $X$ and deterministic
horizon $T\in(0,\infty)$ define
\begin{align}
R & :=\left\{ (H\cdot X)_{T}\mid H\mbox{ admissible}\right\} ,\label{Eq:Replicable_Claims}\\
C & :=\left\{ g\in L^{\infty}(\Omega,\cF_{T},P)\mid g\le f\mbox{ for some }f\in R\right\} .\label{Eq:Super_Replicable_Claims}
\end{align}
The condition \emph{no free lunch with vanishing risk (NFLVR) }with
respect to $X$ on horizon $T$ is
\begin{align*}
\bar{C}\cap L_{+}^{\infty}(\Omega,\cF_{T},P) & =\{0\},
\end{align*}
where $L_{+}^{\infty}$ denotes the a.s. bounded nonnegative random
variables, and $\bar{C}$ is the closure of $C$ with respect to the
norm topology of $L^{\infty}(\Omega,\cF_{T},P)$.
\end{defn}
The following FTAP characterizes NFLVR when $X$ is a general piecewise
semimartingale. When $X$ has more regularity we can and do say more
below. We will need the notion of an \emph{equivalent supermartingale
measure} \emph{(ESMM)} for $\cV(X)$: a measure equivalent to $P$
under which every $V\in\cV(X)$ is a supermartingale.
\begin{thm}[FTAP]
 \label{Thm:FTAP}Let $X$ be a  piecewise semimartingale and $T\in(0,\infty)$.
$X$ satisfies NFLVR on horizon $T$ if and only if there exists an
E$\sigma$MM for $X^{T}$, if and only if there exists an ESMM for
$\cV(X^{T})$.\end{thm}
\begin{proof}
First the equivalence $\mbox{NFLVR}\Longleftrightarrow\mbox{ESMM}$
is proved, and then $\mbox{ESMM}\Longleftrightarrow\mbox{E\ensuremath{\sigma}MM}$.

($\mbox{NFLVR}\imply\mbox{ESMM}$) The implication holds via the main
result of Kabanov \cite[Theorems 1.1 and 1.2]{Art:Kabanov:OnFTAPKrepsDelbSchach:1997}.
To apply his result, we need that the subset of $\R$-valued semimartingales
\begin{align*}
\fG^{1} & :=\left\{ (H\cdot X)\mid H\mbox{ is predictable, }X\mbox{-integrable, and }H\cdot X\ge-1\right\}
\end{align*}
is closed in the semimartingale topology, which is provided by Proposition
\ref{Prop:Memin}. The other technical conditions needed to apply
Kabanov's theorem are straightforward from Proposition \ref{Prop:StochIntChar}.

($\mbox{ESMM}\imply\mbox{NFLVR}$) Suppose that $Q$ is a measure
that makes $V$ a supermartingale for each $V\in\cV(X^{T})$, and
that $(H^{j})$ is a sequence of $X$-admissible trading strategies
such that $\lim_{j\to\infty}(H^{j}\cdot X)_{T}=\psi\ge0$, with convergence
in $L^{\infty}$-norm. For each $j\in\N$ there exists a $v_{j}>0$
such that $V^{v_{j},H^{j}}=v_{j}+(H^{j}\cdot X)\ge0$, so each $(H^{j}\cdot X)^{T}$
is a $Q$-supermartingale. Therefore, $E^{Q}[(H^{j}\cdot X)_{T}]\le0$,
for each $j\in\N$. But $E^{Q}[\cdot]$ is a continuous linear functional
on $L^{\infty}$. Hence, $E^{Q}[\psi]\le0$, which implies that $\psi=0$
a.s.

($\mbox{E\ensuremath{\sigma}MM}\imply\mbox{ESMM}$) If $Q$ is an
E$\sigma$MM for $X^{T}$, then $H\cdot X^{T}$ is a $Q$-sigma-martingale
for all $X^{T}$-admissible $H$. It is also a $Q$-supermartingale,
since any sigma-martingale uniformly bounded from below is a supermartingale
\cite[p. 216]{Book:JacodShiryaev:LimitThmsStochProc:2002}.

($\mbox{ESMM}\imply\mbox{E\ensuremath{\sigma}MM}$) Let $\tilde{Q}$
be an ESMM for $\cV(X^{T})$, equivalently for $\cV(X)$ on horizon
$T$, and let $\tilde{Z}:=\frac{\d\tilde{Q}}{\d P}\in\cF_{T}$. Define
$\tilde{Z}^{k}:=\frac{E[Z\mid\cF_{\tau_{k}}]}{E[Z\mid\cF_{\tau_{k-1}}]}$,
so that $\tilde{Z}=\prod_{k=1}^{\infty}\tilde{Z}^{k}$, with convergence
in $L_{1}$. If $H^{k,n}$ is $X^{k,n}$-admissible, then $H:=H^{k,n}\uI_{(\tau_{k-1},\tau_{k}]\cap\R_{+}\times\Omega^{k,n}}+0^{(N)}\in\cL_{0}(X)$,
and $H\cdot X=H^{k,n}\cdot X^{k,n}$, so is $X$-admissible. Hence,
$C(X^{k,n},\bF)\subseteq C(X,\bF)$, and so $\tilde{Z}$ is an ESMM
for $\cV(X^{k,n})$ on horizon $T$. Since $H^{k,n}\in\cL(X^{k,n})$
implies $(H^{k,n}\cdot X^{k,n})_{T}\in\cF_{\tau_{k}\wedge T}$, and
$(H^{k,n}\cdot X^{k,n})$ takes the value $0$ on $[0,\tau_{k-1}]$,
then $\tilde{Z}^{k}$ is the Radon-Nikodym derivative for an ESMM
for $\cV(X^{k,n})$, $\forall k,n\in\N$. By \cite[Proposition 4.7]{Art:DelbSchach:FundThmMathFin:1998},
$\forall k,n\in\N$, $\forall\varepsilon>0$, there exist E$\sigma$MMs
for $X^{k,n}$, generated by $Z_{\varepsilon}^{k,n}$ that satisfy
$E[|Z_{\varepsilon}^{k,n}-\tilde{Z}^{k}|]<\varepsilon2^{-n}$. The
$Z_{\varepsilon}^{k,n}$ may be assumed to satisfy $E[Z_{\varepsilon}^{k,n}\mid\cF_{\tau_{k-1}\wedge T}]=1$,
since the $\tilde{Z}^{k}$ satisfy this. Then $Z_{\varepsilon}^{k}:=\I_{\{\tau_{k-1}=\infty\}}+\sum_{n=1}^{\infty}\I_{\Omega^{k,n}}Z_{\varepsilon}^{k,n}$
generates $Q^{k}$, an E$\sigma$MM for $X^{k,n}$, $\forall k,n\in\N$,
and $E[|Z_{\varepsilon}^{k}-\tilde{Z}^{k}|]<\varepsilon$.

The process $\hat{Z}_{\varepsilon}^{1}:=\tilde{Z}(Z_{\varepsilon}^{1}/\tilde{Z}^{1})=Z_{\varepsilon}^{1}\prod_{j=2}^{\infty}\tilde{Z}_{j}$
satisfies $\lim_{\varepsilon\to0}\hat{Z}_{\varepsilon}^{1}=Z$ in
probability, and $E[\hat{Z}_{\varepsilon}^{1}]=1=E[Z]$, $\forall\varepsilon>0$.
Therefore, $\lim_{\varepsilon\to0}\hat{Z}_{\varepsilon}^{1}=Z$ in
$L_{1}$. Supposing that $\hat{Z}_{\varepsilon}^{k}$ is defined,
and $\lim_{\varepsilon\to0}\hat{Z}_{\varepsilon}^{k}=Z$ in $L_{1}$.
Then as above we have that $\forall\delta_{k}>0,$ $\exists\varepsilon_{k}>0$
such that $\hat{Z}^{k+1}:=\hat{Z}^{k}(Z_{\varepsilon_{k}}^{k+1}/\tilde{Z}^{k+1})$
satisfies $E[|\hat{Z}^{k+1}-\hat{Z}^{k}|]<\delta_{k}$. So we induce
that there exists a sequence $(\varepsilon_{k})$ such that $(\hat{Z}^{k})$
is a Cauchy sequence in $L_{1}$. Hence, there exists an $L_{1}$-limit,
$Z=\tilde{Z}\prod_{k=1}^{\infty}(Z_{\varepsilon_{k}}^{k}/\tilde{Z}^{k})=\prod_{k=1}^{\infty}Z_{\varepsilon_{k}}^{k}$,
and $E[Z]=1$. We henceforth drop the subscripts $\varepsilon_{k}$
from $Z^{k}$. Since $(\prod_{j=1}^{k}Z^{j},\cF_{\tau_{k}\wedge T})_{k\in\N}$
is a martingale, closed by $Z$, then the convergence is a.s. as well.
For almost all $\omega$, $\tau_{k}(\omega)\nearrow\infty$, so there
exists a $k_{\omega}$ such that $Z^{k}(\omega)=1$ for all $k>k_{\omega}$,
and $Z(\omega)=\prod_{k=1}^{k_{\omega}}Z^{k}(\omega)$. Since $Z^{k}(\omega)>0$,
$\forall k$, then $Z(\omega)>0$. Therefore, $Z=\frac{\d Q}{\d P}$
for a measure $Q\sim P$.

It remains to show that $Q$ generated by $Z$ is a sigma-martingale
measure for $X$. Below, we will show that $X^{k,n}$ is a $Q$-sigma-martingale,
$\forall k,n\in\N$. Supposing this, an application of Proposition
\ref{Prop:X_Sig_Mart_iff_X^k,n_Sig_Mart} implies that $H\cdot(\sum_{n=1}^{\infty}\uI_{\Omega^{k,n}}X^{k,n}+0^{(N)})=\sum_{n=1}^{\infty}H^{k,n}\cdot X^{k,n}$
is a $Q$-sigma-martingale, for any $H\in\cL(X)$, $\forall k\in\N$.
Then $(H\cdot X)^{\tau_{m}}=H_{0}^{\prime}X_{0}+\sum_{k=1}^{m}\sum_{n=1}^{\infty}H\cdot X^{k,n}$
is a $Q$-sigma-martingale, so $H\cdot X$ is locally a $Q$-sigma-martingale,
and is thus a $Q$-sigma-martingale. This proves that $X$ is a piecewise
$Q$-sigma-martingale by definition.

($X^{k,n}$ is a $Q$-sigma-martingale) Define the martingale $Z_{t}:=E[Z\mid\cF_{t}]$,
$0\le t\le T$, and let $M^{k,n}$ be an $\R^{n}$-valued $Q^{k}$-martingale
($\frac{\d Q^{k}}{\d P}:=Z^{k}$), such that $X^{k,n}=H^{k,n}\cdot M^{k,n}$
for some $H^{k,n}\in\cL(M^{k,n})$. Such processes $M^{k,n}$ and
$H^{k,n}$ always exist, since $X^{k,n}$ is a $Q^{k}$-sigma-martingale.
Since $X^{k,n}=0^{(n)}$ on the complement of $(\tau_{k-1},\infty)\cap\R_{+}\times\Omega^{k,n}$,
we may choose $M^{k,n}=0^{(n)}$ on this set. Since $Z^{k,n}X^{k,n}$$=(Z^{k,n}X^{k,n})^{\tau_{k}}$,
then we may replace $M^{k,n}$ with $(M^{k,n})^{\tau_{k}}$. Then
we have for $0\le s\le t\le T$,
\begin{align*}
E[Z_{t}M_{t}^{k,n}\mid\cF_{s}] & =E[E(Z_{t}(M_{t}^{k,n})^{\tau_{k}}\mid\cF_{s\vee\tau_{k}})\mid\cF_{s}],\\
 & =E\biggl[Z_{t}^{\tau_{k}}M_{t}^{k,n}E\biggl(\frac{Z_{t}}{Z_{t}^{\tau_{k}}}\mid\cF_{s\vee\tau_{k}}\biggr)\mid\cF_{s}\biggr],\\
 & =\frac{Z_{s}}{Z_{s}^{\tau_{k}}}E[Z_{t}^{\tau_{k}}M_{t}^{k,n}\mid\cF_{s}],\\
 & =\frac{Z_{s}}{Z_{s}^{\tau_{k}}}E[E(Z_{t}^{\tau_{k}}M_{t}^{k,n}\mid\cF_{s\vee\tau_{k-1}})\mid\cF_{s}],\\
 & =\frac{Z_{s}}{Z_{s}^{\tau_{k}}}E\Bigl[\prod_{j=1}^{k-1}Z_{t}^{j}E\Bigl(\sum_{m=1}^{\infty}\I_{\Omega^{k,m}}Z_{t}^{k,m}M_{t}^{k,n}\mid\cF_{s\vee\tau_{k-1}}\Bigr)\mid\cF_{s}\Bigr],\\
 & =\frac{Z_{s}}{Z_{s}^{\tau_{k}}}E\Bigl[\prod_{j=1}^{k-1}Z_{t}^{j}E(\I_{\Omega^{k,n}}Z_{t\vee\tau_{k-1}}^{k,n}M_{t\vee\tau_{k-1}}^{k,n}\mid\cF_{s\vee\tau_{k-1}})\mid\cF_{s}\Bigr],\\
 & =\frac{Z_{s}}{Z_{s}^{\tau_{k}}}Z_{s}^{k,n}M_{s}^{k,n}E\Bigl[\prod_{j=1}^{k-1}Z_{t}^{j}\mid\cF_{s}\Bigr],\\
 & =Z_{s}M_{s}^{k,n},
\end{align*}
where we made use of $M^{k,n}=0^{(n)}$ on the complement of $(\tau_{k-1},\infty)\cap\R_{+}\times\Omega^{k,n}$,
and that $Z_{t}^{k,n}M_{t}^{k,n}=Z_{t\vee\tau_{k-1}}^{k,n}M_{t\vee\tau_{k-1}}^{k,n}$.
This establishes that $ZM^{k,n}$ is a $P$-martingale, so $ZX^{k,n}=Z(H^{k,n}\cdot M^{k,n})$
is a $P$-sigma-martingale, and thus $X^{k,n}$ is a $Q$-sigma-martingale.
\end{proof}
The proof of Theorem \ref{Thm:FTAP} contains the proof of the following
corollary, the generalization of \cite[Proposition 4.7]{Art:DelbSchach:FundThmMathFin:1998}
and \cite[Theorem 2]{Art:Kabanov:OnFTAPKrepsDelbSchach:1997}.
\begin{cor}
If $\tilde{Q}$ is an ESMM for $\cV(X)$, then for any $\varepsilon>0$
there exists $Q$, an E$\sigma$MM for $X$, such that $Q$ and $\tilde{Q}$
are within $\varepsilon$ of each other with respect to the total
variation norm.
\end{cor}
When $X$ is a bounded $\R^{n}$-valued semimartingale, then any ESMM
for $\cV(X)$ is an equivalent martingale measure (EMM) for $X$,
since $-X_{i},X_{i}\in\cV(X)$, for $1\le i\le n$. However, in the
piecewise setting, even if $X$ is $\R^{n}$-valued and satisfies
NFLVR, nevertheless $\left|X\right|_{1}$ bounded is not sufficient
regularity for existence of an EMM for $X$. This can be seen by reconsidering
Example \ref{Ex:Bounded_Piecewise_Strict_Local_Mart}, where a bounded
piecewise strict local martingale is constructed. In lieu of the
$|X|_{1}$ bounded assumption, we have the following sufficient condition
for existence of an EMM for $X$.
\begin{cor}
If each simple, predictable $H$ with $\left|H\right|_{1}$ bounded
is admissible for $X$, then any ESMM for $\cV(X)$ is an EMM for
$X$. Therefore, in this case, NFLVR for $X$ is equivalent to the
existence of an EMM for $X$. \textbf{}\end{cor}
\begin{proof}
Under the premise, if $H\in\bS(N)$, $H$ is predictable, and $\left|H\right|_{1}$
is bounded, then $H$ and $-H$ must both be admissible. This means
that $H\cdot X$ and $-H\cdot X$ are both $Q$-supermartingales,
hence $Q$-martingales. Thus $X$ is a piecewise $Q$-martingale.
\end{proof}
The following corollary is the natural generalization of the FTAP
for $\R^{n}$-valued locally bounded semimartingales, proved originally
in \cite{Art:DelbSchach:FundThmMathFin:1994}.
\begin{cor}
\label{Cor:FTAP_Loc_Bdd}If $X$ is locally bounded, then any ESMM
for $\cV(X)$ is an ELMM for $X$. Therefore, in this case, NFLVR
for $X$ is equivalent to the existence of an ELMM for $X$.\end{cor}
\begin{proof}
Let $(\beta_{i})\nearrow\infty$ be a sequence of stopping times
such that $|\I_{\{\beta_{i}>0\}}X^{\beta_{i}}|_{1}$ is bounded, $\forall i\in\N$.
Suppose that $S$ is bounded, $\R^{n}$-simple, and predictable. The
process $H:=S\uI_{(\tau_{k-1},\tau_{k}]\cap\R_{+}\times\Omega^{k,n}}+0^{(N)}$
satisfies $H\in\bS(N)$, $|H|_{1}$ is bounded, and $H\cdot X=S\cdot X^{k,n}$.
We have $-S,S\in\cL(X^{k,n})$, and $S\cdot(X^{k,n})^{\beta^{i}}$
is bounded, so there exists $v>0$ such that $(v+S\cdot(X^{k,n})^{\beta^{i}})$,
$(v+((-S)\cdot(X^{k,n})^{\beta^{i}}))\in\cV$, making them both $Q$-supermartingales
for any ESMM $Q$ for $\cV$. Therefore, $S\cdot(X^{k,n})^{\beta^{i}}$
is a $Q$-martingale, making $(X^{k,n})^{\beta^{i}}$ also a $Q$-martingale.
Hence, $X^{k,n}$ is a $Q$-local martingale, $\forall k,n\in\N$,
and Proposition \ref{Prop:X_Loc_Mart_iff_X^k,n_Loc_Mart} implies
that $X$ is a piecewise  $Q$-local martingale.
\end{proof}
It is straightforward that if NFLVR holds for $X$, then it holds
for each $X^{k,n}$, but we state the result formally for completeness.
\begin{cor}
\label{Cor:FTAP_for_X_Implies_For_All_X_k_n}If $X$ satisfies NFLVR
on deterministic horizon $T\in(0,\infty)$, then for any reset sequence
$(\tau_{k})$ and for each $k,n\in\N$, $X^{k,n}$ satisfies NFLVR
on horizon $T$. \end{cor}
\begin{proof}
If $H^{k,n}$ is $X^{k,n}$-admissible, then $H:=H^{k,n}\uI_{(\tau_{k-1},\tau_{k}]\cap\R_{+}\times\Omega^{k,n}}+0^{(N)}\in\cL(X)$,
and $H\cdot X=H^{k,n}\cdot X^{k,n}$, implying that $H$ is $X$-admissible.
Therefore, $C(X^{k,n},\bF)\subseteq C(X,\bF)$, and any FLVR with
respect to $X^{k,n}$ is a FLVR with respect to $X$.
\end{proof}
The converse of Corollary \ref{Cor:FTAP_for_X_Implies_For_All_X_k_n}
is of course false in general: Let $X$ be any $\R^{n}$-valued semimartingale
such that on all horizons $T\in(0,\infty)$ $\mbox{NA}_{1}$ holds
and NFLVR fails. For example, see the log-pole singularity market
of \cite{Art:FernKaratzKard:DiversityAndRelArb:2005,Art:Karatzas&Fernholz:SPTReview:2009}.
Then by Theorem \ref{Thm:NA1_ELMD}, there exists an equivalent local
martingale deflator $Z$ for $\cV(X^{T})$. Any fundamental sequence
$(\tau_{k}$) for $Z$ is a reset sequence for $X$, so that for any
$k\in\N$, $Z_{T}^{\tau_{k}}$ is a Radon-Nikodym derivative for an
ELMM, and thus ESMM, for $\cV(X^{\tau_{k}\wedge T})$, and thus for
$\cV((X^{k,n})^{T})$. By Theorem \ref{Thm:FTAP}, NFLVR holds for
$X^{k,n}$ on horizon $T$, $\forall k,n\in\N$.

The following corollary gives a sufficient additional criterion to
yield the converse implication.
\begin{cor}
For reset sequence $(\tau_{k})$, if $K_{T}:=\sum_{k=1}^{\infty}\I_{T\ge\tau_{k}}\le\kappa\in\N$,
and NFLVR holds on horizon $T\in(0,\infty)$ for $X^{k,n}$ for each
$n\in\N$, $k\le\kappa$, then it holds for $X$.\end{cor}
\begin{proof}
An ELMD $Z:=\prod_{k=1}^{\kappa}\tilde{Z}^{k}$ for $\cV(X^{T})$
may be constructed as in the proof of Theorem \ref{Thm:FTAP}, where
$\tilde{Z}^{k}$ generates an E$\sigma$MM for $\sum_{n=1}^{\infty}\hat{\I}_{\Omega^{k,n}}X^{k,n}+0^{(N)}$.
Then $Z_{T}^{\tau_{\kappa+1}}=Z_{T}$, so $E[Z_{T}]=E[Z_{0}]=1$,
proving closure. The sigma-martingality of $(XZ)_{0\le t\le T}$ follows
from the same arguments as given in the proof of Theorem \ref{Thm:FTAP}.
\end{proof}

\section{Concluding remarks\label{Secc:Conclusion}}

The notions of $\R^{n}$-valued semimartingale, martingale, and relatives,
may be extended by localization to a  piecewise semimartingale of
stochastic dimension, et al. The stochastic integral $H\cdot X$ may
be extended in kind, by pasting together pieces of stochastic integrals
from $\R^{n}$-valued segments. This construction seems to preserve
nearly all of the properties of $\R^{n}$-stochastic analysis that
are local in nature. Care is needed when extending results relying
on the boundedness of processes, as this notion is ambiguous in $\cup_{n=1}^{\infty}\R^{n}$.
However, the notion of local boundedness is unambiguous. Some properties
that are not local in nature extend as well, such as the FTAP of Delbaen
and Schachermayer.

Piecewise semimartingale models open up the possibility of studying
more realistic and varied market dynamics, for example, allowing companies
to enter, leave, merge and split in an equity market. This has already
seen application, in extending the results presented in \cite{Art:WinslowFouque:RegDivArb:2010}
to the more realistic setting in \cite{Diss:Winslow}.

The generalization of the various forms of the fundamental theorem
of asset pricing suggests that many of the results \cite{Book:DelbSchach:ArbBook:2006,Book:FollmerSchied:Stoch_Fin_Discrete_Time:2004}
pertaining to super-replication and hedging that exploit $\sigma(L_{\infty},L_{1})$-duality
may also extend to the  piecewise setting. We leave investigation
of these and other properties to future work.

\bigskip{}

\noindent \begin{flushleft}
\textbf{Acknowledgements: }I would like to thank Jean-Pierre Fouque
for discussions and feedback.
\par\end{flushleft}

\appendix
\begin{appendices}

\section{\label{App:Proofs}Proof\emph{ }of Proposition\emph{ }\ref{Prop:Well_Definition_of_Semis_and_Stoch_Int}}

Let $X$ be progressive and have paths with left and right limits
for all times. Let $(\tau_{k})$ be a reset sequence such that $X^{k,n}$
is an $\R^{n}$-valued semimartingale for each $k,n\in\N$. Let $\cL(X)$
and $H\cdot X$ be defined with respect to $(\tau_{k})$. Suppose
that $(\tilde{\tau}_{k})$ is an arbitrary reset sequence for $X$,
with corresponding $\tilde{X}^{k,n}$, $\tilde{H}^{k,n}$, $\tilde{\Omega}^{k,n}$,
$\tilde{\cL}(X)$. For $H\in\cL_{0}(X)$, we have
\begin{align}
H\cdot X & =\sum_{j=1}^{\infty}[(H\cdot X)^{\tilde{\tau}_{j}}-(H\cdot X)^{\tilde{\tau}_{j-1}}],\nonumber \\
 & =\sum_{j,k,n=1}^{\infty}\left[\left(H^{k,n}\cdot X^{k,n}\right)^{\tilde{\tau}_{j}}-\left(H^{k,n}\cdot X^{k,n}\right)^{\tilde{\tau}_{j-1}}\right],\nonumber \\
 & =\sum_{j,k,n=1}^{\infty}(H^{k,n}\I_{(\tilde{\tau}_{j-1},\tilde{\tau}_{j}]})\cdot((X^{k,n})^{\tilde{\tau}_{j}}\I_{(\tilde{\tau}_{j-1},\infty)}),\nonumber \\
 & =\sum_{j,k,n=1}^{\infty}((H\uI_{(\tau_{k-1}\vee\tilde{\tau}_{j-1},\tau_{k}\wedge\tilde{\tau}_{j}]\cap\R_{+}\times\Omega^{k,n}}+0^{(n)}))\nonumber \\
 & \quad\cdot((X-X_{\tau_{k-1}}^{+})^{\tau_{k}\wedge\tilde{\tau}_{j}}\uI_{(\tau_{k-1}\vee\tilde{\tau}_{j-1},\infty)\cap\R_{+}\times\Omega^{k,n}}+0^{(n)}),\nonumber \\
 & =\sum_{j,k,n=1}^{\infty}((H\uI_{(\tau_{k-1}\vee\tilde{\tau}_{j-1},\tau_{k}\wedge\tilde{\tau}_{j}]\cap\R_{+}\times\Omega^{k,n}}+0^{(n)}))\label{Eq:Semi+Int_Well_Def:1}\\
 & \quad\cdot(X^{\tau_{k}\wedge\tilde{\tau_{j}}}\uI_{[\tau_{k-1}\vee\tilde{\tau}_{j-1},\infty)\cap\R_{+}\times\Omega^{k,n}}+0^{(n)}),\nonumber \\
 & =\sum_{j,k,n=1}^{\infty}((H\uI_{(\tau_{k-1}\vee\tilde{\tau}_{j-1},\tau_{k}\wedge\tilde{\tau}_{j}]\cap\R_{+}\times\tilde{\Omega}^{j,n}}+0^{(n)}))\label{Eq:Semi+Int_Well_Def:2}\\
 & \quad\cdot(X^{\tau_{k}\wedge\tilde{\tau}_{j}}\uI_{[\tau_{k-1}\vee\tilde{\tau}_{j-1},\infty)\cap\R_{+}\times\tilde{\Omega}^{j,n}}+0^{(n)}),\nonumber \\
 & =\sum_{k,j,n=1}^{\infty}((H\uI_{(\tau_{k-1}\vee\tilde{\tau}_{j-1},\tau_{k}\wedge\tilde{\tau}_{j}]\cap\R_{+}\times\tilde{\Omega}^{j,n}}+0^{(n)}))\label{Eq:Semi+Int_Well_Def:3}\\
 & \quad\cdot((X-X_{\tilde{\tau}_{j-1}})^{\tau_{k}\wedge\tilde{\tau}_{j}}\uI_{(\tau_{k-1}\vee\tilde{\tau}_{j-1},\infty)\cap\R_{+}\times\tilde{\Omega}^{j,n}}+0^{(n)}),\nonumber \\
 & =\sum_{k,j,n=1}^{\infty}(\tilde{H}^{j,n}\I_{(\tau_{k-1},\tau_{k}]})\cdot((\tilde{X}^{j,n})^{\tau_{k}}\I_{(\tau_{k-1},\infty)}),\nonumber \\
 & =\sum_{k,j,n=1}^{\infty}\left[\left(\tilde{H}^{j,n}\cdot\tilde{X}^{j,n}\right)^{\tau_{k}}-\left(\tilde{H}^{j,n}\cdot\tilde{X}^{j,n}\right)^{\tau_{k-1}}\right],\nonumber \\
 & =\sum_{j,n=1}^{\infty}\left[\left(\tilde{H}^{j,n}\cdot\tilde{X}^{j,n}\right)\right].\label{Eq:Semi+Int_Well_Def:4}
\end{align}
The steps utilize only definitions and basic properties of $\R^{n}$-valued
stochastic analysis. Equations (\ref{Eq:Semi+Int_Well_Def:1}) and
(\ref{Eq:Semi+Int_Well_Def:3}) follow because the integrand is zero
when the shift in the integrator takes effect. Equation (\ref{Eq:Semi+Int_Well_Def:2})
follows from $(\tau_{k-1}\vee\tilde{\tau}_{j-1},\tau_{k}\wedge\tilde{\tau}_{j}]\cap\R_{+}\times\Omega^{k,n}=(\tau_{k-1}\vee\tilde{\tau}_{j-1},\tau_{k}\wedge\tilde{\tau}_{j}]\cap\R_{+}\times\tilde{\Omega}^{j,n}$.

To prove that $\tilde{X}^{j,n}$ is a semimartingale $\forall j,n\in\N$,
suppose that the $\R^{n}$-valued simple predictable processes $(S^{i})_{i\in\N}$
and $S$ satisfy $\lim_{i\to\infty}S^{i}=S$, with the convergence
being ucp (assumed throughout). Then for $j,n\in\N$, define $H^{i}:=S^{i}\uI_{(\tilde{\tau}_{j-1},\tilde{\tau_{j}}]\cap\R_{+}\times\tilde{\Omega}^{j,n}}+0^{(N)}$,
so that $\lim_{i\to\infty}H^{i}=H:=S\uI_{(\tilde{\tau}_{j-1},\tilde{\tau_{j}}]\cap\R_{+}\times\tilde{\Omega}^{k,n}}+0^{(N)}$.
Since each $S^{i}$ is $\R^{n}$-simple predictable, then each $H^{k,m,i}$,
formed by dissecting $H^{i}$ as in (\ref{Eq:H^kn-Def}), is $\R^{m}$-simple
predictable, and therefore $H^{i}\in\cL(X)$. By Proposition \ref{Prop:Semis_Are_Cont_Integrators},
$\lim_{i\to\infty}H^{i}\cdot X=H\cdot X$, and by (\ref{Eq:Semi+Int_Well_Def:4}),
$H\cdot X=(S\uI_{(\tilde{\tau}_{j-1},\tilde{\tau_{j}}]\cap\R_{+}\times\tilde{\Omega}^{j,n}}+0^{(N)})\cdot X=S\cdot\tilde{X}^{j,n}$.
Since $H^{i}\cdot X=S^{i}\cdot\tilde{X}^{j,n}$, this proves that
$\lim_{i\to\infty}S^{i}\cdot\tilde{X}^{j,n}=S\cdot\tilde{X}^{j,n}$,
and therefore $\tilde{X}^{j,n}$ is a semimartingale.

Equation (\ref{Eq:Semi+Int_Well_Def:4}) above shows that $H\in\cL(X)\imply H\in\tilde{\cL}(X)$,
and furthermore that $H\cdot X=\wt{H\cdot X}$. The reset sequences
$(\tau_{k})$ and $(\tilde{\tau}_{k})$ are arbitrary, so $\cL(X)=\tilde{\cL}(X)$,
and $H\cdot X$ is independent of the choice of reset sequence.\qed

\end{appendices}

\bibliographystyle{spmpsci}


\appendix

\end{document}